\documentclass[pdftex,11pt,svgnames]{article}

\usepackage[utf8]{inputenc}            
\usepackage{amsmath}
\usepackage{amssymb}  
\usepackage{amsbsy} 
\usepackage{listings}
\usepackage{graphics}

\usepackage[multiple]{footmisc}

\usepackage{enumerate}
\usepackage[shortlabels]{enumitem}

\IfFileExists{dsfont.sty}{
	\usepackage{dsfont}
	\usepackage[usenames,dvipsnames]{color} 
	
}{
	
}

\usepackage[numbers]{natbib}

\usepackage{pdfpages}
\usepackage{hyperref} 
\hypersetup{
    unicode=false,          
    pdftoolbar=true,        
    pdfmenubar=true,        
    pdffitwindow=false,     
    pdfstartview={FitH},    
    pdftitle={Computation of bonus in multi-state life insurance},    
    pdfauthor={Jamaal Ahmad, Kristian Buchardt, and Christian Furrer},     
    pdfsubject={Life Insurance Mathematics},   
    pdfcreator={ },   
    pdfproducer={ }, 
    pdfkeywords={ } { } { }, 
    pdfnewwindow=true,      
    colorlinks=true,       
    linkcolor=Black, 
    citecolor=Black,        
    filecolor=Black,      
    urlcolor=Black 
}

\usepackage[font=small,format=plain,labelfont=bf,up,textfont=it,up]{caption}

\usepackage[a4paper]{geometry}  
\usepackage[english]{babel}

\usepackage[autolanguage,np]{numprint}

\usepackage{doi}
\usepackage{float}

\usepackage{fancyhdr}
\usepackage{amsthm}
\theoremstyle{plain}
\newtheorem{Theorem}{Theorem}
\newtheorem{theorem}[Theorem]{Theorem}

\newtheorem{assumption}[Theorem]{Assumption}

\newtheorem{proposition}[Theorem]{Proposition}
 
\newtheorem{corollary}[Theorem]{Corollary} 

\newtheorem{definition}[Theorem]{Definition}

\theoremstyle{definition}

\newtheorem{example}[Theorem]{Example}

\newcommand\xqed[1]{%
  \leavevmode\unskip\penalty9999 \hbox{}\nobreak\hfill
  \quad\hbox{#1}}
\newcommand\demo{\xqed{$\circ$}}
\newcommand\demoo{\xqed{$\triangle$}}

\theoremstyle{remark}

\newtheorem{remark}[Theorem]{Remark}

\numberwithin{Theorem}{section}

\numberwithin{equation}{section}

\newcommand{\md}{{\, \mathrm{d} }}

\newcommand{\deriv}[2]{\frac{\md {#1}}{\md {#2}}}

\newcommand{\tuborg}[1]{\{ #1 \}}

\DeclareMathOperator{\E}{E}

\DeclareMathOperator{\pr}{P}

\DeclareMathAlphabet{\mathpzc}{OMS}{pzc}{m}{it}
\newcommand{\J}{\mathpzc{J}}

\newcommand{\Jp}{\J^\mathrm{p}}
\newcommand{\Jf}{\J^\mathrm{f}}

\makeatletter

\newcommand{\Rmnum}[1]{\expandafter\@slowromancap\romannumeral #1@}
\makeatother

\usepackage{wasysym} 

\usepackage{tikz}
\usetikzlibrary{arrows,positioning,calc, external} 

\tikzset{
    >=stealth',
    punkt/.style={
           rectangle,
           rounded corners,
           draw=black, thick,
           text width=7em,
           minimum height=2em,
           text centered},
    punktl/.style={
           re
           tangle,
           rounded corners,
           draw=black, thick,
           
           text width=7em,
           minimum height=2em,
           text centered},
    pil/.style={
           ->,
           shorten <=4pt,
           shorten >=4pt,},
    pildotted/.style={
           ->,
           shorten <=4pt,
           shorten >=4pt,
  dotted,
  },
  external/system call={pdflatex \tikzexternalcheckshellescape 
                                        -halt-on-error
                                        -interaction=batchmode 
                                        -jobname "\image" "\texsource"
                                        && pdftops -eps "\image.pdf"}
}
\usepackage{subfig}
\usepackage{todonotes}

\usepackage{authblk}

\makeatletter
\let\@fnsymbol\@arabic
\makeatother

\title{Computation of bonus in multi-state life insurance}
\author[1,$\star$]{Jamaal Ahmad}
\author[2]{Kristian Buchardt}
\author[1,2]{Christian Furrer}
\affil[1]{\footnotesize Department of Mathematical Sciences, University of Copenhagen, Universitetsparken 5, DK-2100 Copenhagen \O, Denmark.}
\affil[2]{\footnotesize PFA Pension, Sundkrogsgade 4, DK-2100 Copenhagen \O, Denmark.}
\affil[$\star$]{\footnotesize Corresponding author. E-mail: \href{mailto:jamaal@math.ku.dk}{jamaal@math.ku.dk}.}

\date{ }

\begin{document}
\maketitle{}

\addtocounter{footnote}{4} 

\begin{center}
{\sc Abstract}
\end{center}
{\small

We consider computation of market values of bonus payments in multi-state with-profit life insurance. The bonus scheme consists of additional benefits bought according to a dividend strategy that depends on the past realization of financial risk, the current individual insurance risk, the number of additional benefits currently held, and so-called portfolio-wide means describing the shape of the insurance business. We formulate numerical procedures that efficiently combine si-\linebreak mulation of financial risk with classic methods for the outstanding insurance risk. Special attention is given to the case where the number of additional benefits bought only depends on the financial risk. Methods and results are illustrated via a numerical example.

\vspace{5mm}


\vspace{5mm}



}

 	 	 	
\section{Introduction}\label{sec:intro}

The potential of systematic surplus in multi-state with-profit life insurance (sometimes referred to as participating life insurance) leads to bonus payments that depend on the development of the financial market and the states of the insured. This dependence is typically non-linear and involves the whole paths of the processes governing the financial market and the states of the insured. Consequently, the computation of market values of bonus payments lies outside the scope of classic backward and forward methods. In this paper, we present computational schemes for a selection of these more involved market values using a combined approach in which we simulate the financial risk while retaining classic analytical methods and numerical methods for differential equations in regards to the outstanding insurance risk.

In the preexisting literature on valuation of both with-profit and equity/unit-linked life insurance, the focus is prevalently on financial risk, so that it is not uncommon to disregard biometric and behavioral risks more or less completely, see e.g.~\cite{Bauer2006,Bauer2008} and, concerning with-profit life insurance, cf.~the discussion in~\cite[][Section 1]{JensenSchomacker2015}. Exceptions include for instance~\cite{Bacinello2001,mollersteffensen} and, more recently, \cite{BacinelloMillossovichChen2018}, where both financial risk and mortality/longevity risk are considered. Our emphasis is in the spirit of~\cite{JensenSchomacker2015} and concerns the specific challenges that universally arise from including event risk via multi-state modeling. This places our research as part of the literature on multi-state modeling for with-profit life insurance. For a non-technical introduction to multi-state modeling in life insurance, see~\cite{Koller2012}.

For Danish with-profit products, the investment strategy and dividend strategy are to a great extent controlled by the insurer, and practitioners have traditionally determined the market value of bonus payments residually, cf.~\cite[][Chapter 2]{mollersteffensen}. This is achieved by considering the available assets together with the market value of guaranteed payments and imposing the equivalence principle on the market basis in combination with certain ad hoc adjustments. In case parts of the assets correspond to expected future profits, also denoted as contractual service margins in the IFRS 17 regulatory framework, cf.~\cite{ifrs17}, the equivalence principle is invalidated, which points to more sophisticated computational methods. The provision of these kinds of methods constitutes the main contribution of this paper.

We distinguish between dividends and bonuses:
The development of the portfolio and the financial market typically gives rise to a surplus,
which is distributed among the policyholders' policies via dividend yields according to the chosen dividend strategy.
After this allocation, the accumulated dividends on each policy are paid out according to a specific bonus scheme,
and we refer to these extra payments as the bonus payments.
The study of systematic surplus, dividends and bonus payments in multi-state with-profit life insurance goes back to \cite{Ramlau1991,NorbergBonus2,NorbergBonus},
where one finds careful definitions of various concepts of surplus,
discussions of general principles for its redistribution,
and the introduction of forecasting techniques in a so-called Markov chain interest model,
see also~\cite{Norberg1995}.
In~\cite{Steffensen2006}, partial differential equations for market values of so-called predetermined payments and bonus payments are derived in a Black-Scholes model. 

The projection of bonus payments in multi-state life insurance and the computation of associated market values has recently received renewed attention, see~\cite{JensenSchomacker2015,Jensen2016,Lollike2020,Falden2020}. In~\cite{Jensen2016}, the focus is on projection of bonus payments conditionally on the insured sojourning in a specific state; this approach targets e.g.\ product design and bonus prognosis from the perspective of the insured rather than market valuation. Conversely, the paper~\cite{JensenSchomacker2015} also deals with projection of bonus payments but on a portfolio level, which ensures computational feasibility but does not shed light on the full complexity of multi-state with-profit life insurance. Although with-profit life insurance focuses on the portfolio of insured and although decisions by the insurer (so-called future management actions), including possible determination of dividend yields, often depend mainly on the performance of said portfolio, one ought to take into account that bonus payments in principle are  allocated to the individual insured. This is the starting point in~\cite{Lollike2020}, where the focus is on deriving differential equations for relevant retrospective reserves given a dividend strategy (used to buy additional benefits) that depends in an affine manner on the reserves themselves. The process governing the state of the insured is assumed Markovian. In~\cite{Falden2020}, the results of~\cite{Lollike2020} are extended to allow for policyholder behavior, namely the options of surrender and free policy conversion. The surrender option allows the policyholder to cancel all future payments and instead receive a single payment corresponding in some sense to the value of the contract, while the free policy option allows the policyholder to cancel future premiums at the cost of reducing future benefits. In~\cite{Lollike2020,Falden2020}, the dependence of the dividend strategy on the financial state of the insurance business, encapsulated in what we below shall term the \textit{shape} of the insurance business, and the practical and computational challenges arising from this are not highlighted nor studied. This paper derives its main novelty value from addressing these challenges within a multi-state framework while also allowing for financial risk.

In this paper, we derive methods for the computation of market values of bonus payments in a Markovian multi-state model for a financial market consisting of one risky asset in addition to a bank account governed by a potentially stochastic interest rate. The insurance risk and financial risk are assumed independent. We include incidental modeling of the policyholder options surrender and free policy conversion following~\cite{henriksen2014,buchardt2015,BuchardtMollerSchmidt}. In regards to dividends and bonus, we adopt a somewhat universal dividend strategy; examples throughout the text, including the numerical example, provide links to the preexisting literature and actuarial practice by focusing on dividends arising via a second order interest rate, which in particular gives rise to an interest rate guarantee. The allocated dividends are later paid out according to a bonus scheme, and we here focus solely on the bonus scheme known as \textit{additional benefits}, where dividends are simply used to buy extra benefits; this bonus scheme constitutes the focal point of~\cite[][Chapter 6]{mollersteffensen} and is also quite common in practice.

In practice, the dividend strategy depends on product design, regulatory frameworks, and decisions made by the insurer. In this paper, we assume that the dividend strategy is explicitly computable based on the following information: the past realization of financial risk, the current individual insurance risk (state of insured and time since free policy conversion), the current  shape of the insurance business, and the number of additional benefits currently held. Furthermore, the dividend strategy must be affine in the number of additional benefits. The shape of the insurance business consists first and foremost of so-called portfolio-wide means, cf.~\cite[][Chapter 6]{mollersteffensen}, which reflect on a portfolio level the current financial state of the insurance business and thus are relevant to the insurer in determining the dividend strategy and investment strategy. In particular, the shape of the insurance business includes portfolio-wide means of the technical reserve of guaranteed payments and of the expected accumulated guaranteed cash flows. Consequently, the shape of the insurance business depends on the dividend strategy, which again depends on the shape of the insurance business.

Using classic techniques, we derive a system of differential and integral equations for the computation of the expected accumulated bonus cash flows conditionally on the realization of financial risk. This allows us to formulate a procedure for the computation of the market value of bonus payments which efficiently combines simulation of financial risk with classic methods for the remaining insurance risk. By not needing to simulate insurance risk, our procedure has a significant advantage compared to full-blown Monte Carlo methods. We identify the special case where the number of additional benefits depend only on financial risk -- the \textit{state-independent} case -- and show how this significantly simplifies the numerical procedure. It is our impression that the state-independent model is aligned to current actuarial practice, where it might e.g.\ serve as an approximation for valuation on a portfolio level. This is further examined in a numerical example.

We should like to stress that while our results are subject to important technical regularity conditions, it is the general methodology and conceptual ideas that constitute the main contributions of this paper. Furthermore, our concepts, methods, and results are targeted academics and actuarial practitioners alike, and, consequently, we aim at keeping the presentation at a reasonable technical level.

The paper is structured as follows. In Section~\ref{sec:setup}, we present the setup. The general results and general numerical procedure are given in Section~\ref{sec:projection}, while the state-independent case is the subject of Section~\ref{sec:state_indep_projection}. Section~\ref{sec:num_ex} is devoted to a numerical example. Finally, Section~\ref{sec:outlook} concludes with a comparison to recent advances in the literature and a discussion of possible extensions.

\section{Setup}\label{sec:setup}

In the following, we describe the mathematical framework. Subsections~\ref{subsec:prelim}--\ref{subsec:unit_val} introduce the processes governing the financial market, the state of the insured, and the insu\-rance payments, and we discuss the valuation of so-called predetermined payments. The\linebreak dividend and bonus scheme is described in Subsection~\ref{subsec:div_bon}, which leads to a specification of the total payment stream as a sum of predetermined payments and bonus payments. Contrary to the predetermined payments, the bonus payments depend on the development of the financial market, which adds an extra layer of complexity to the valuation problem. The focal point of this paper is to establish explicit methods for the computation of the market value of the bonus payments; a precise description of this problem is given in Subsection~\ref{subsec:liabs}. In the remainder of the paper, the problem is studied for a specific class of dividend processes specified in Subsection~\ref{subsec:shape_ctrl}.

A background probability space $(\Omega,\mathbb{F},\pr)$ is taken as given. Unless explicitly stated or evident from the specific context, all statements are in an almost sure sense w.r.t.\ $\pr$. The probability measure $\pr$ relates to market valuation and therefore corresponds to some risk neutral probability measure. Due to the presence of insurance risk, the market is not complete, which implies that the risk neutral probability measure is not unique. Since we shall assume financial risk and insurance risk to be independent, one can think of the probability measure $\pr$ as the product measure of some risk neutral probability measure for financial risk and some probability measure for insurance risk. 

\subsection{Preliminaries} \label{subsec:prelim}

The state of the insured is governed by a non-explosive jump process  $Z=\tuborg{Z(t)}_{t\geq 0}$ on a finite state space $\J$ with deterministic initial state $Z(0)\equiv z_0 \in \J$. Denote by $N$ the corresponding multivariate counting process with components $N_{jk}=\tuborg{N_{jk}(t)}_{t\geq 0}$ for $j,k\in \J,  k\neq j$ given by
\begin{equation*}
N_{jk}(t)=\# \tuborg{s\in (0,t] : Z(s-)=j, \ Z(s)=k}.
\end{equation*}
Let $S_1=\{S_1(t)\}_{t\geq 0}$ be the price process for some risky asset (diffusion process, in particular continuous) and let $r=\{r(t)\}_{t\geq 0}$ be a suitably regular short rate process with corresponding bank account $S_0(t) = S_0(0)\exp\left(\int_0^t r(v)\md v\right)$, $S_0(0) \equiv s_0 > 0$, and suitably regular forward interest rates $f(t,\cdot)$, $t\geq 0$, satisfying
\begin{align*}
\E\!\left[\left. e^{-\int_t^T r(s) \md s} \, \right| \mathcal{F}^S(t)\right]
=
e^{-\int_t^T f(t,s) \md s}
\end{align*}
for all $0\leq t<T$ as well as $f(t,t) = r(t)$ for all $t\geq 0$; here $\mathcal{F}^S$ is the natural filtration generated by $S=(S_0,S_1)$, which exactly represents available market information. The available insurance information is represented by the filtration $\mathcal{F}^Z$ naturally generated by $Z$, and the total information available is represented by the filtration $\mathcal{F} = \mathcal{F}^S \vee \mathcal{F}^Z$ naturally generated by $(S,Z)$.

To allow for free policy behavior and surrender, we suppose the state space $\J$ can be decomposed as
\begin{align*}
\J = \Jp\cup \Jf,
\end{align*}
with $\Jp := \{0,\ldots,J\}$ and $\Jf := \{J+1,\ldots,2J+1\}$ for some $J\in\mathbb{N}$. Here $\Jp$ contains the premium paying states, while $\Jf$ contains the free policy states, and transition to $\{J\}$ and $\{2J+1\}$ corresponds to surrender as premium paying and free policy, respectively, cf.~\cite{buchardt2015,BuchardtMollerSchmidt}.\ We suppose that $\Jf$ is absorbing and can only be reached via a transition from $\{0\}$ to $\{J+1\}$, $\{J\}$ and $\{2J+1\}$ are absorbing, and that $\{J\}$ and $\{2J+1\}$ can only be reached from $\{0\}$ and $\{J+1\}$, respectively. The setup is depicted in Figure~\ref{fig:extmm}.
\begin{figure}[h!]
\centering
\scalebox{0.9}{
\begin{tikzpicture}[node distance=4em, auto] 
	\node[punkt] (g) {$0$};
	\node[right = of  g] (dummy) {};
	\node[text width=2em,minimum height=1em, text centered, below=2em of dummy] (i) {$\cdots$};
	\node[punkt, right=of dummy] (i1) {$i$};
	\node[punkt, below = 3em of i] (dead) {$J-1$}; 
	\path 	($(g.south)+(0.4,-0.2)$) edge [pildotted, bend right=10] ($(i.west)+(-0.1,0)$)
        	($(i.west)+(0,0.0)$) edge [pildotted, bend right=10]  ($(g.south)+(0.5,-0.2)$)
		($(i1.south)$) edge [pil, bend right=10] (dead.north east)
        	($(dead.north east)+(0.1,0)$) edge [pil, bend right=10]  ($(i1.south)+(0.1,0)$)
		($(i.east)+(0.1,0)$) edge [pildotted, bend right=10] ($(i1.south)+(-0.4,-0.2)$)
        	($(i1.south)+(-0.6,-0.2)$) edge [pildotted, bend right=10]  ($(i.east)+(-0.1,0)$)
		($(dead.north)$) edge [pildotted, bend right = 10] ($(i.south)$)
		($(i.south)+(-0.1,0)$) edge [pildotted, bend right = 10] ($(dead.north)+(-.1,0)$)
		($(dead.north west)$) edge [pil, bend right = 10] ($(g.south)$)
		($(g.south)+(-0.1,0)$) edge [pil, bend right = 10] ($(dead.north west)+(-0.1,0)$)
		($(g.east)+(0,-0.1)$) edge [pil, bend right = 10] ($(i1.west)+(0,-0.1)$)
		($(i1.west)$) edge [pil, bend right = 10] ($(g.east)$)
	;
	\node[right=10pt] at ($(i1.south)+(0,-2.4)$) {$\mathcal{J}^\mathrm{p}$};
	
	\node[punkt, left = 5em of g] (J) {$J$ \\ Surrender};
	\path	(g.west) edge [pil] (J.east);

	\draw[thick, dotted] ($(J.north west)+(-0.8,0.8)$) rectangle ($(i1.south east)+(0.8,-3.5)$);

	\node[punkt, below = 13 em of g] (J1) {$J+1$};
	\node[right = of J1] (dummy2) {};
	\node[text width=2em,minimum height=1em, text centered, below= 2em of dummy2] (Ji) {$\cdots$};
	\node[punkt, right = of dummy2] (Ji1) {$J+1+i$};
	\node[punkt, below = 3em of Ji] (2J) {$2J$};	
	\path ($(g.south)+(-.2,0)$) edge [pil] ($(J1.north)+(-.2,0)$);	 
	\path 	($(J1.south)+(0.4,-0.2)$) edge [pildotted, bend right=10] ($(Ji.west)+(-0.1,0)$)
        	($(Ji.west)+(.0,0)$) edge [pildotted, bend right=10]  ($(J1.south)+(0.5,-0.2)$)
		($(Ji1.south)$) edge [pil, bend right=10] (2J.north east)
        	($(2J.north east)+(.1,0)$) edge [pil, bend right=10]  ($(Ji1.south)+(.1,0)$)
		($(Ji.east)+(0.1,0)$) edge [pildotted, bend right=10] ($(Ji1.south)+(-0.4,-0.2)$)
        	($(Ji1.south)+(-0.6,-0.2)$) edge [pildotted, bend right=10]  ($(Ji.east)+(-0.1,0)$)
		($(2J.north)$) edge [pildotted, bend right = 10] ($(Ji.south)$)
		($(Ji.south)+(-0.1,0)$) edge [pildotted, bend right = 10] ($(2J.north)+(-0.1,0)$)
		($(2J.north west)$) edge [pil, bend right = 10] ($(J1.south)$)
		($(J1.south)+(-0.1,0)$) edge [pil, bend right = 10] ($(2J.north west)+(-0.1,0)$)
		($(J1.east)+(0,-0.1)$) edge [pil, bend right = 10] ($(Ji1.west)+(0,-0.1)$)
		($(Ji1.west)$) edge [pil, bend right = 10] ($(J1.east)$)
	;
	\node[right=10pt] at ($(Ji1.south)+(0,-2.4)$) {$\mathcal{J}^{\mathrm{f}}$};

	\node[punkt, left = 5em of J1] (2J1) {$2J+1$ \\ Surrender \\ as free policy};
	\path	(J1.west) edge [pil] (2J1.east);
	\draw[thick, dotted] ($(2J1.north west)+(-0.8,0.5)$) rectangle ($(Ji1.south east)+(0.8,-3.5)$);
\end{tikzpicture}}
\caption{General finite state space extended with a surrender state $\{J\}$ and free policy states $\Jf$.\ The states $\Jp\setminus\tuborg{J}$ contain the biometric states of the insured, e.g.\ active, disabled, and dead.\ The states $\mathcal{J}^\mathrm{f}$ are a copy of $\mathcal{J}^p$, and a transition from $\{0\}$ to $\{J+1\}$ corresponds to a free policy conversion. A transition to $\{J\}$ or $\{2J+1\}$ corresponds to a surrender of the policy.}
\label{fig:extmm}
\end{figure}
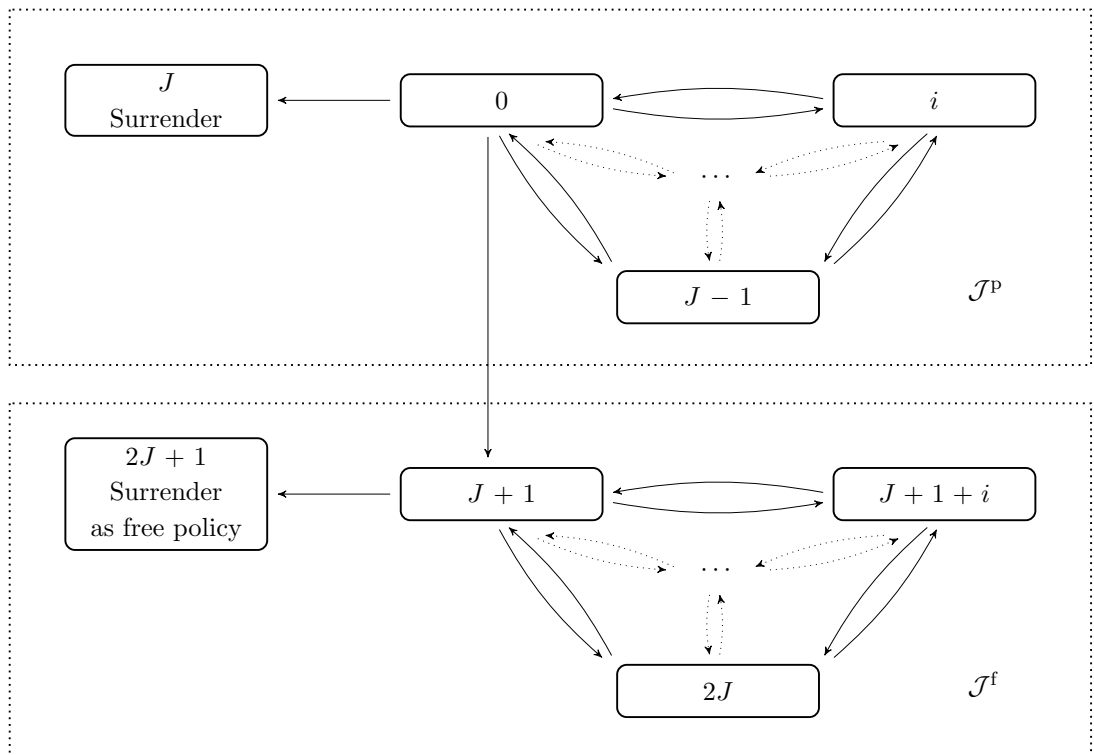

\subsection{Life insurance contract with policyholder options} \label{subsec:contract}
The life insurance contract is described by a payment stream $B=\tuborg{B(t)}_{t\geq0}$\linebreak giving accumulated benefits less premiums. It consists of \textit{predetermined payments}\linebreak $B^\circ=\tuborg{B^\circ(t)}_{0\leq t \leq n}$, stipulated from the beginning of the contract, and additional bonus payments determined when market and insurance information are realized during the course of the contract; details regarding the latter are given in later subsections.

We specify the predetermined payments as in~\cite{buchardt2015,BuchardtMollerSchmidt}. For simplicity, we suppose that the predetermined payments regarding the classic states $\Jp$ consist of suitably regular deterministic sojourn payment rates $b_j$ and transition payments $b_{jk}$; in particular, surrender results in a deterministic payment. In the free policy states, no premiums are paid and the benefit payments are reduced by a factor $\rho \in [0,1]$ depending on the time of free policy conversion. In rigorous terms, we have
\begin{align*}
\md B^\circ(t)
&=
\md B^{\circ,\text{p}}(t)
+
\rho(\tau)\md B^{\circ,\text{f}}(t), \quad &&B^\circ(0)=0, \\
\md B^{\circ,\text{p}}(t)
&=
\sum_{j\in \Jp} \mathds{1}_{(Z(t-)=j)}\bigg(b_j(t)\md t+\sum_{k\in \Jp\atop k \neq j} b_{jk}(t) \md N_{jk}(t)\bigg), \quad &&B^{\circ,\text{p}}(0)=0, \\
\md B^{\circ,\text{f}}(t)
&=
b_{(J+1)(2J+1)}(t) \md N_{(J+1)(2J+1)}(t) \\
&\quad+ 
\sum_{j\in \Jf} \mathds{1}_{(Z(t-)=j)}\bigg(b_{j'}(t)^+\md t+ \hspace{-5mm}\sum_{k\in \Jf\setminus\{2J+1\}\atop k \neq j} \hspace{-5mm} b_{j'k'}(t)^+ \md N_{jk}(t)\bigg), \quad &&B^{\circ,\text{f}}(0)=0, 
\end{align*}
with $\Jf \ni j \mapsto j' := j-(J+1)$, $x^+:=\max\{0,x\}$, and $x^- := -\min\{0,x\}$ (for $x\in\mathbb{R}$), and where $\tau$ is the time of free policy conversion given by
\begin{align*}
\tau = \inf\{t \in [0,\infty) : Z(t) \in \Jf\}.
\end{align*}
We have $\tau = 0$ if and only if $z_0 \in \Jf$; in this case, the policy is initially a free policy. Without loss of generality we thus let $\rho(0) = 1$. Furthermore, we suppose there are no sojourn payments in the surrender states, i.e.\ $b_J \equiv 0$. 

For later purposes, we also define a benefit payment stream $B^{\circ,+}$ mimicking $B^{\circ,\text{f}}$, but extended to all of $\mathcal{J}$:
\begin{align*}
B^{\circ,+}(t)
&=
B^{\circ,\text{p},+}(t)
+
B^{\circ,\text{f}}(t), \\
\md B^{\circ,\text{p},+}(t)
&=
b_{(J+1)(2J+1)}(t) \md N_{0J}(t) \\
&\quad+ 
\sum_{j\in \Jp} \mathds{1}_{(Z(t-)=j)}\bigg(b_j(t)^+\md t+ \hspace{-3mm}\sum_{k\in \Jp\setminus\{J\}\atop k \neq j} \hspace{-3mm} b_{jk}(t)^+ \md N_{jk}(t)\bigg), \quad B^{\circ,\text{p},+}(0)=0. 
\end{align*}
In the following, we assume the existence of a maximal contract time $n\in(0,\infty)$ in the sense that all sojourn payment rates and transition payments, including those of the unit bonus payment stream, cf.~Subsection~\ref{subsec:div_bon}, are zero for $t>n$.

\subsection{Valuation of the predetermined payments} \label{subsec:unit_val}

The life insurance contract is written on the \textit{technical basis}, also called the first order basis, which is, at least originally, designed to consist of prudent assumptions on financial risk and insurance risk. The technical basis is modeled via another probability measure $\pr^\star$ under which the short rate process $r^\star$ is deterministic and suitably regular, while $Z$ is independent of $S$ and Markovian with suitably regular transition rates $\mu^\star$. The assumptions regarding absorption, as illustrated in Figure~\ref{fig:extmm}, are retained under $\pr^\star$. Policyholder behavior is typically not explicitly included on the technical basis; in particular, there is no change in technical reserve upon free policy conversion. Following along the lines of~\cite{buchardt2015,BuchardtMollerSchmidt}, this is consistent with the assumption that the transition rates under the technical basis, the surrender payments, and the free policy factor take the form
\begin{align}\label{eq:tech_spec}
\begin{split}
&\mu_{jk}^\star = \mu_{j'k'}^\star, \hspace{30mm} j,k \in \Jf, k \neq j, \\
&b_{0J} = \widetilde{V}_0^\star, \\
&b_{(J+1)(2J+1)} = \widetilde{V}^{\star,+}_0, \\
&(0,\infty) \ni t \mapsto \rho(t) =
\frac{
\widetilde{V}^\star_0(t)}
{\widetilde{V}^{\star,+}_0(t)},
\end{split}
\end{align}
where for $j \in \Jp\setminus\{J\}$ the state-wise technical reserve $\widetilde{V}^\star_j$ of predetermined payments and the corresponding valuation of benefits only $\widetilde{V}^{\star,+}$ are given by
\begin{align}\label{eq:TechnicalTildeCirc}
\widetilde{V}^\star_j(t) &= \E^{\star}\!\!\left[\left. \int_t^n e^{-\int_t^s r^\star(v)\md v} \md B^{\circ}(s)\, \right|  Z(t) = j\right]\!, \\
\widetilde{V}^{\star,+}_j(t) &= \E^{\star}\!\!\left[\left. \int_t^n e^{-\int_t^s r^\star(v)\md v} \md B^{\circ,+}(s)\, \right|  Z(t) = j\right]\!,
\end{align}
with $\E^\star$ denoting integration w.r.t.\ $\pr^\star$. The specification of the technical reserve is circular, since the payments depend on the technical reserve, which again depends on the payments.The setup described here is closely related to the generic situation described in \cite[][Section 4.3]{ChristiansenDjehiche2020}, where circularity is resolved by a careful construction of the involved processes and using backward stochastic differential techniques; for instance, our specification~\eqref{eq:tech_spec} ensures that equation (4.8) in~\cite{ChristiansenDjehiche2020} holds.
 
It is possible to show that the state-wise technical reserves of predetermined payments satisfy the following differential equations of Thiele type:
\begin{align}\label{eq:thiele_tech}
\frac{\md}{\md t} \widetilde{V}^\star_j(t) = r^\star(t) \widetilde{V}^\star_j(t) - b_j(t) -\hspace{-3mm} \sum_{k \in \Jp\setminus\{J\}\atop k \neq j} \hspace{-3mm} \big(b_{jk}(t) + \widetilde{V}^\star_k(t) - \widetilde{V}^\star_j(t)\big) \mu_{jk}^\star(t), \hspace{2mm} \widetilde{V}^\star_j(n) = 0,
\end{align}
for $j \in \Jp\setminus\{J\}$. By adding $+$'s as superscripts, one finds an identical system of differential equations concerning the valuation of benefits only.

We are now ready to define the technical reserve of predetermined payments denoted $V^{\star,\circ}$. First, for the purpose of bonus allocation, the definitions of state-wise reserves of predetermined payments are naturally extended from $j \in \Jp\setminus\{J\}$ to $j \in \J$ via
\begin{align} \label{eq:TechnicalUnit}
V_j^{\star,\circ}(t)
=
\begin{cases} 
\widetilde{V}^\star_j(t) & \text{ if } j \in \Jp\setminus\{J\}, \\
\rho(\tau) \widetilde{V}^{\star,+}_{j'}(t) & \text{ if } j \in \Jf\setminus\tuborg{2J+1}, \\
0 & \text{ if } j \in \tuborg{J, 2J+1}.
\end{cases}
\end{align}

Note that $V_j^{\star,\circ}$ depends on $\tau$ in the free policy states, thus being stochastic, while it is deterministic in the premium paying states. The technical reserve of predetermined payments $V^{\star,\circ}$ is now defined according to $V^{\star,\circ}(t)=V^{\star,\circ}_{Z(t)}(t)$. In particular, $V^{\star,\circ}(t)$ is a version of the conditional expectation
\begin{align*}
\E^\star\!\left[\left. \int_t^n e^{-\int_t^s r(u) \md u} \, \md B^\circ(s) \, \right| \mathcal{F}(t)
\right]\!.
\end{align*}
In practice, the technical reserves are exactly computed according to~\eqref{eq:thiele_tech} and~\eqref{eq:TechnicalUnit}. Here we provide a probablistic setup and specification that is consistent with this approach.

We now turn our attention to valuation under the market basis modeled via $\pr$. Here we assume that $Z$ and $S$ are independent and that $Z$ is Markovian with suitably regular transition rates $\mu$. The market reserve $V^\circ$ of predetermined payments is then given by
\begin{align}\label{eq:Vcirc}
V^\circ(t)
&=
\E\!\left[\left. \int_t^n e^{-\int_t^s r(u) \md u} \, \md B^\circ(s) \, \right| \mathcal{F}(t)
\right] = \int_t^n e^{-\int_t^s f(t,u)\md u} A^{\circ}(t,\mathrm{d}s),
\end{align}
with $A^\circ$ the so-called expected accumulated predetermined cash flows given by
\begin{align}\label{eq:unit_cf_def}
A^\circ(t,s)
=
\E\!\left[\left. B^\circ(s) - B^\circ(t) \, \right| \mathcal{F}^Z(t)\right].
\end{align}
Denote with $p$ the transition probabilities of $Z$ under $\pr$. Following~\cite{buchardt2015, BuchardtMollerSchmidt}, on $(Z(t) \in \Jf)$,
\begin{align}\label{eq:UnitCFjf}
\begin{split}
A^\circ(t,\mathrm{d}s)
=
\rho(\tau)\Bigg(&p_{Z(t)(J+1)}(t,s)\widetilde{V}^{\star,+}_0(s) \mu_{(J+1)(2J+1)}(s) \md s \\
&+
\sum_{j\in\Jf} p_{Z(t)j}(t,s)\bigg(b_{j'}(s)^+ + \hspace{-5mm} \sum_{k\in\Jf\setminus\{2J+1\}\atop k \neq j} \hspace{-5mm} b_{j'k'}(s)^+ \mu_{jk}(s)\bigg)\mathrm{d}s\Bigg),
\end{split}
\end{align}
while on $(Z(t)\in\Jp)$, 
\begin{align}\label{eq:UnitCFjp}
\begin{split}
A^\circ(t,\mathrm{d}s)
=& \,
\sum_{j\in\Jp} p_{Z(t)j}(t,s)\bigg(b_j(s) + \sum_{k\in\Jp\atop k \neq j} b_{jk}(s) \mu_{jk}(s)\bigg)\mathrm{d}s \\
&+p^\rho_{Z(t)(J+1)}(t,s)\widetilde{V}^{\star,+}_0(s) \mu_{(J+1)(2J+1)}(s) \md s \\
&+\sum_{j\in\Jf} p^\rho_{Z(t)j}(t,s)\bigg(b_{j'}(s)^+ + \hspace{-5mm} \sum_{k\in\Jf\setminus\{2J+1\}\atop k \neq j} \hspace{-5mm} b_{j'k'}(s)^+ \mu_{jk}(s)\bigg)\mathrm{d}s
\end{split}
\end{align}
where the so-called $\rho$-modified transition probabilities $p_{jk}^\rho$, $j\in\Jp$ and $k\in\J$, are defined by $p_{jk}^\rho(t,s) = \E[\mathds{1}_{(Z(s)=k)} \rho(\tau)^{\mathds{1}_{(\tau\leq s)}}\, | \, Z(t) = j]$ and satisfy for $k\in \Jf$ so-called $\rho$-modified versions of Kolmogorov's forward differential equations:
\begin{align}\label{eq:kol_mod}
\begin{split}
\deriv{}{s} p_{jk}^\rho(t,s)
&=
 \sum_{\ell \in \Jf\atop \ell \neq k} p^\rho_{j\ell}(t,s) \mu_{\ell k}(s) + \mathds{1}_{(k=J+1)}p_{j0}(t,s)\mu_{0k}(s)\rho(s) - p_{jk}^\rho(t,s) \mu_{k\bullet }(s), \\
p_{jk}^\rho(t,t) &= 0,
\end{split}
\end{align}
while $p_{jk}^\rho(t,s) = p_{jk}(t,s)$ for $k\in \Jp$.

\subsection{Dividends and bonus}\label{subsec:div_bon}

With premiums determined by the principle of equivalence based on the prudent technical basis,  the portfolio creates a systematic surplus if everything goes well. This surplus mainly belongs to the insured and is to be paid back in the form of dividends. Following \cite{NorbergBonus2, NorbergBonus}, we let $D=\tuborg{D(t)}_{t\geq 0}$ denote the accumulated dividends, and we suppose it only consists of absolutely continuous dividend yields:
\begin{equation*}
 \md D(t)=\delta(t)\md t, \quad D(0)=0,
 \end{equation*} 
where $\delta=\tuborg{\delta(t)}_{t\geq 0}$ is suitably regular and $\mathcal{F}$-adapted. In Subsection~\ref{subsec:shape_ctrl}, we specify the dividend strategy further. Classic examples include dividends distributed via a second order interest rate, see Example~\ref{ex:rente} below.

The dividends are allocated to the individual life insurance contract but not yet paid out; payout occurs at a possibly later point in time according to some specific bonus scheme. In the following, we adopt the bonus scheme known as \textit{additional benefits}, that is we suppose that the dividends are used as a premium to buy additional benefits on the technical basis corresponding to a so-called unit bonus payment stream $B^\dagger$ that only consists of benefits and thus is unaffected by the free policy option. It is given by
\begin{align*}
\md B^\dagger(t) &= \sum_{j\in \J} \mathds{1}_{(Z(t-)=j)}\bigg(b_j^\dagger(t)\md t + \sum_{k\in \J \atop k\neq j} b_{jk}^\dagger(t)\md N_{jk}(t)\bigg), \quad B^\dagger(0)=0, 
\end{align*} 
where the payment functions in the premium paying states $\Jp$, $b_j^\dagger$ and $b_{jk}^\dagger$, are suitably regular non-negative deterministic functions with $b_J^\dagger\equiv0$, while
\begin{align*}
&b^{\dagger}_j =b^{\dagger}_{j'} \quad \text{and} \quad b^{\dagger}_{jk} =b^{\dagger}_{j' k'}, \hspace{10mm} j,k \in \Jf, k \neq j, \\
&b_{0J}^{\dagger} = \widetilde{V}^{\star,\dagger}_0,
\end{align*}
where for $j\in\Jp\setminus\{J\}$ we denote by $\widetilde{V}_j^{\star,\dagger}$ the state-wise technical unit reserves of $B^\dagger$ given as \eqref{eq:TechnicalTildeCirc} with $B^\circ$ replaced by $B^\dagger$.\ Again, these state-wise technical reserves satisfy differential equations of Thiele type, namely~\eqref{eq:thiele_tech} with added superscripts $\dagger$.

For the purpose of bonus allocation, the state-wise technical unit reserves are naturally extended from $j \in \Jp\setminus\{J\}$ to $j \in \J$ via
\begin{align}\label{eq:TechnicalTilde}
V_j^{\star,\dagger}(t)
=
\begin{cases} 
\widetilde{V}^{\star,\dagger}_{j}(t) & \text{ if } j \in \Jp\setminus\{J\}, \\
\widetilde{V}^{\star,\dagger}_{j'}(t) & \text{ if } j \in \Jf\setminus\tuborg{2J+1}, \\
0 & \text{ if } j \in \tuborg{J, 2J+1},
\end{cases}
\end{align}
when the technical value of the additional benefits $V^{\star,\dagger}$ reads $V^{\star,\dagger}(t) = V_{Z(t)}^{\star,\dagger}(t)$.

The expected accumulated unit bonus cash flows $A^{\dagger}$ of $B^\dagger$ on the market basis can be found analogously to $A^\circ$ and read
\begin{align}\label{eq:CFtilde}
A^{\dagger}(t, \md s) &= a^{\dagger}(t,s) \md s, \\ \label{eq:CFtilde_a}
a^{\dagger}(t,s) &= \sum_{j\in \J}p_{Z(t)j}(t,s)\bigg(b^\dagger_j(s)+\sum_{k\in \J\atop k\neq j} b^\dagger_{jk}(s)\mu_{jk}(s)\bigg).
\end{align}
The state-wise counterparts are denoted $A_i^{\dagger}$ and $a_i^{\dagger}$, $i \in \J$. They satisfy $A_{Z(t)}^{\dagger}(t,\md s) = a_{Z(t)}^{\dagger}(t,s) \md s = a^{\dagger}(t,s) \md s = A^{\dagger}(t,\md s)$ by taking the form
\begin{align}\label{eq:CFtilde_state}
A_i^{\dagger}(t, \md s) &= a_i^{\dagger}(t,s) \md s, \\
\label{eq:CFtilde_state_a}
a_i^{\dagger}(t,s) &=
\sum_{j\in \J}p_{ij}(t,s)\bigg(b^\dagger_j(s)+\sum_{k\in \J\atop k\neq j} b^\dagger_{jk}(s)\mu_{jk}(s)\bigg).
\end{align}
The market unit reserve of $B^\dagger$ denoted $V^\dagger$ is given analogously to~\eqref{eq:Vcirc} but with the superscript $\dagger$ replacing $\circ$. Note that contrary to the market reserve of predetermined payments, it may be taken to only depend on $Z(t)$ at time $t$, since this is the case for the corresponding expected accumulated unit bonus cash flow. We write $V^\dagger_{Z(t)}(t)$ in place of $V^\dagger(t)$ to highlight this fact.

Having specified the unit bonus payment stream $B^\dagger$ as well as its technical value and associated expected accumulated cash flows, we are now in a position to finalize the specification of the bonus scheme. Thus, let $Q(t)$ denote the number of additional benefits, that is the number of benefit streams $B^\dagger$, held at time $t$. Since $\delta$ is used as a premium to buy $B^\dagger$ on the technical basis, we have that (cf.\ equation (4.11) in~\cite{NorbergBonus})
\begin{align}\label{eq:Q}
\md Q(t)
=
\frac{\mathrm{d}D(t)}
{V_{Z(t)}^{\star,\dagger}(t)}
=
\frac{\delta(t)}
{V_{Z(t)}^{\star,\dagger}(t)}
\md t, \quad Q(0)=0. 
\end{align}
Imposing this bonus mechanism, the total payment stream consisting of both predetermined payments and bonus payments is given by
\begin{align}\label{eq:B}
\md B(t)
&=
\md B^{\circ}(t)
+
Q(t)\md B^\dagger(t), \quad B(0) = 0.
\end{align}
In this paper, we implicitly think of $Q$ as weakly increasing, although this is not a mathematical requirement. This way of thinking is reflected in the terminology. Along these lines, we define the payment process $B^g$ by
\begin{align}\label{eq:Bg}
B^g(t,\mathrm{d}s)&=\mathrm{d} B^\circ(s)+Q(t)\md B^\dagger(s), \quad B^g(t,t)=B(t),
\end{align}
and refer to it as the \textit{payments guaranteed at time $t\geq0$}, while the remaining payments 
\begin{align*}
    \left(Q(s)-Q(t)\right)\!\md B^\dagger(s)
\end{align*}
are referred to as \textit{bonus (payments)}.

The number of unit bonus payment streams $B^\dagger$ held increases according to~\eqref{eq:Q}, where $V^{\star, \dagger}_{Z(t)}(t)$ is the price (on the technical basis) at time $t$ of a unit bonus payment stream. Thus, at time $t$, the guaranteed payment stream originating from bonus is increased to $s \mapsto Q(t)\md B^\dagger(s)$. The technical reserve of the contract, which is formally defined in~\eqref{eq:technical} below, is the value on the technical basis of the combined benefits given by~\eqref{eq:Bg}. By construction, this technical reserve increases exactly by the increase in accumulated dividends, whereas this is typically not the case for the market value; the market value is formally defined in~\eqref{eq:Vg} and~\eqref{eq:CF} below.

In the remainder of the paper, we focus on valuation of the payment stream \eqref{eq:B}, in particular the bonus payments. We assume that $Q$ exists and is suitably regular, so that the technical arguments in the remainder of the paper are legitimate. This is an implicit condition that must be checked for any specific model.

\subsection{Liabilities} \label{subsec:liabs}

Thinking of time zero as now, the present life insurance liabilities of the insurer are described by the market value of the total payment stream $B$ evaluated at time zero:
\begin{align*}
V(0)=\E\!\left[\int_0^n e^{-\int_0^t r(v)\md v}\md B(t)\right]\!.
\end{align*}
By \eqref{eq:B}, this amounts to market valuation of the predetermined payments and bonus payments. Thus $V(0) = V^\circ(0)+V^b(0)$ where $V^\circ(0)$ is given by \eqref{eq:Vcirc} and  
\begin{align}\label{eq:Vb0}
V^b(0)=\E\!\left[\int_0^n e^{-\int_0^t r(v)\md v} Q(t)\md B^\dagger(t)\right]\!.
\end{align}
is the time zero market value of bonus payments.

\begin{remark}\label{rmk:Qneq0}
By setting $Q(0)=0$, we think of time zero as the time of initialization of the insurance contract. To determine the market value of bonus payments after initialization of the contract, one could extend the filtration $\mathcal{F}$ to include additional information at time zero and consider a general $\mathcal{F}(0)$-adapted $Q(0)$. This extension is straightforward and achieved by focusing on $Q(\cdot)-Q(0)$ rather than $Q(\cdot)$, and thus the requirement $Q(0)=0$ is only really made for notational convenience. \demoo
\end{remark}

There exist well-established methods to calculate $V^\circ(0)$ explicitly using the expected accumulated cash flows of predetermined payments on the market basis from \eqref{eq:UnitCFjf}--\eqref{eq:UnitCFjp}; in particular, this computation does not depend on the dividend strategy $\delta$ nor further realizations of the financial market (only the forward rate curve $f(0,\cdot)$ is required). On the contrary, the time zero market value of bonus payments $V^b(0)$ \textit{does} depend on the strategy $\delta$. Due to possibly non-linear path dependencies regarding both the financial and biometric/behavioral scenarios, this implies that classic computational methods via ($\rho$-modified) Kolmogorov's forward differential equations are not applicable.

The focal point of the paper is to establish methods to calculate the market value of bonus payments $V^b(0)$. We consider an approach that combines simulations of the financial market with analytical methods and numerical methods for differential equations in regards to calculations involving the state of the insured. Everything else being equal, this approach should be numerically superior to a pure simulation approach for which one would simulate both the financial market and the state of the insured. To formalize the main idea, we define what we shall term $Q$-modified transition probabilities (at time $0$) for $j\in \J$ by
\begin{align}\label{def:pQ}
p_{z_0j}^Q(0,t) = \E\!\left[\left. Q(t)\mathds{1}_{(Z(t)=j)} \, \right|  \mathcal{F}^{S}(t)\right]
\end{align}
for all $t\geq 0$. We immediately have the following result:
\begin{proposition}\label{prop:FDB}
The time zero market value of bonus payments is given by 
\begin{align}\label{eq:FDB}
V^b(0) =& \, \E\!\left[\int_0^n e^{-\int_0^t r(v)\md v} A^b(0,\mathrm{d}t) \right]\!, \\ \label{eq:bonusCF}
A^b(0,\mathrm{d}t) =& \, a^b(0,t) \md t, \\
\label{eq:bonusCFa}
a^b(0,t) :=& \, \sum_{j\in \J}p_{z_0j}^Q (0,t)\Big(b^\dagger_j(t)+\sum_{k\in \J \atop k\neq j} b^\dagger_{jk}(t)\mu_{jk}(t)\Big).
\end{align}
Furthermore, if $Q$ is adapted to $\mathcal{F}^S$, then
\begin{align}\label{eq:pQ_split}
p_{z_0j}^Q(0,t) &= Q(t)p_{z_0j}(0,t), \\
\label{eq:BonusCF_stateIndep}
a^b(0,t) &= Q(t)a^\dagger(0,t).
\end{align}
\end{proposition}
\begin{proof}
Since $\tuborg{Q(t)}_{t\geq 0}$ is continuous and adapted, it is predictable. Using martingale techniques, in particular that
\begin{align*}
Q(t) \md N_{jk}(t) - Q(t) \mathds{1}_{(Z_{t-} = j)} \mu_{jk}(t) \md t
\end{align*}
defines a martingale, we find that
\begin{align*}
V^b(0) &= \E\!\Bigg[ \int_0^n e^{-\int_0^t r(v)\md v} \sum_{j\in \J}Q(t)\mathds{1}_{(Z(t-)=j)}\bigg(b_j^\dagger(t) + \sum_{k\in \J \atop k\neq j} b_{jk}^\dagger(t)\mu_{jk}(t)\bigg) \!\md t \Bigg]\!.
\end{align*}
Since there is almost surely at most a finite number of transitions on each compact time interval, we may replace $\mathds{1}_{(Z(t-)=j)}$ by $\mathds{1}_{(Z(t)=j)}$. Using the law of iterated expectations and Fubini's theorem, we conclude that
  \begin{align*}
V^b(0)&=\E\!\Bigg[ \int_0^n e^{-\int_0^t r(v)\md v} \sum_{j\in \J}\E\!\left[\left. \mathds{1}_{(Z(t)=j)}Q(t)\, \right| \mathcal{F}^{S}(t)\right]\bigg(b_j^\dagger(t) + \sum_{k\in \J \atop k\neq j} b_{jk}^\dagger(t)\mu_{jk}(t)\bigg)\!\md t \Bigg] \\    
&= \E\!\Bigg[\int_0^n e^{-\int_0^t r(v)\md v} \sum_{j\in \J} p_{z_0j}^Q(0,t) \bigg(b_j^\dagger(t) + \sum_{k\in \J \atop k\neq j} b_{jk}^\dagger(t)\mu_{jk}(t)\bigg)\!\md t\Bigg] \\
&=\E\!\Bigg[\int_0^n e^{-\int_0^t r(v)\md v} a^b(0,t) \md t\Bigg]\!.
\end{align*}
Furthermore, if $Q$ is $\mathcal{F}^S$-adapted, then the $Q$-modified transition probabilities satisfy
\begin{equation*}
p_{z_0j}^Q(0,t) =  \E\!\left[\left. \mathds{1}_{(Z(t)=j)}Q(t)\, \right| \mathcal{F}^{S}(t)\right] = Q(t)p_{z_0j}(0,t),
\end{equation*}
and thus $a^b(0,t) = Q(t)a^\dagger(0,t)$, cf.~\eqref{eq:CFtilde_a}. 
\end{proof}

\begin{remark}\label{rmk:approx_expression}
Note that by casting $Q$ according to~\eqref{eq:Q}, by interchanging the order of integration, and by using the law of iterated expectations, it is possible to derive the following alternative formula for the time zero market value of bonus payments:
\begin{align*}
V^b(0) = \int_0^n \E\!\left[e^{-\int_0^s r(v)\md v} \delta(s) \frac{V^\dagger_{Z(s)}(s)}{V_{Z(s)}^{\star,\dagger}(s)}\right] \mathrm{d}s.
\end{align*}
From this expression, we see how the time zero market value of bonus payments consists of an accumulation of time zero market values of additional benefits bought at different points in time. When we later compare different scenario-based projection models, this representation of the time zero value of the bonus payments turns out to be quite useful, cf.\ Example~\ref{ex:second_order_int}.\demoo
\end{remark}

Since the so-called expected accumulated bonus cash flow $A^b(0,\cdot)$ is $\mathcal{F}^S$-adapted, the result provides a representation of $V^b(0)$ motivating a computational scheme based on simulation of the financial market. For each simulated financial scenario, we should compute $A^b(0,\cdot)$ explicitly in each scenario, which in general requires computation of of $p^Q_{z_0j}(0,\cdot)$ for all $j\in \J$; this we study in Section \ref{sec:projection}. In the special case where $Q$ is $\mathcal{F}^S$-adapted, it holds that $p^Q_{z_0j}(0,\cdot) = Q(\cdot) p_{z_0j}(0,\cdot)$, and the problem simplifies to a direct calculation of $Q$ that does not involve the biometric/behavioral states, and can essentially be solved by a classic computation of the expected accumulated cash flow $A^\dagger(0,\cdot)$ via Kolmogorov's forward differential equations; this is studied in Section \ref{sec:state_indep_projection}.

As mentioned above, the computation of the expected accumulated bonus cash flow depends on the actual specification of the dividend strategy $\delta$ during the course of the contract, and in practice, this strategy is a control variable that depends on what we refer to as the shape of the insurance business. In the following subsection, we formalize the shape of the insurance business and its corresponding controls, which leads to a specification of a class of dividend strategies.

\subsection{Shape and controls} \label{subsec:shape_ctrl}

We now introduce the shape of the insurance business consisting of key quantities on a portfolio level that the insurer needs at future time points to determine the controls, i.e.\ the dividend strategy and the investment strategy. We only introduce a few key financial indicators, but we believe that our general methodology allows for the implementation of additional shape variables.

To describe the shape of the insurance business, we first consider the liabilities, specifically the technical value and the market value of guaranteed payments on a portfolio level. Recall that the payments $B^g(t,\cdot)$ guaranteed at time $t\geq 0$ take the form~\eqref{eq:Bg}. The market value of guaranteed payments $V^g$ is thus given by 
\begin{align} \label{eq:Vg}
V^g(t)&= \E\!\left[\left. \int_t^n e^{-\int_t^s r(v)\md v} B^g(t,\mathrm{d}s) \, \right| \mathcal{F}(t)\right] =  \int_t^n e^{-\int_t^s f(t,v)\md v} A^g(t, \md s), 
\end{align}
with $A^g$ denoting the expected accumulated guaranteed cash flows, 
\begin{align}
\label{eq:CF}
A^g(t, \mathrm{d}s) &= A^\circ(t, \mathrm{d}s) + Q(t)A^{\dagger}_{Z(t)}(t, \mathrm{d}s).
\end{align}
Similarly, the technical reserve of guaranteed payments is given by 
\begin{align}\label{eq:technical}
 V^\star(t) = V^{\star,\circ}(t)+Q(t)V_{Z(t)}^{\star,\dagger}(t).
\end{align}
The so-called \textit{portfolio-wide means} of $V^\star$, $A^g$, and $V^g$ are now obtained by averaging out the unsystematic insurance risk by applying the law of large numbers w.r.t.\ a collection of independent and comparable insured in the portfolio, see e.g.\ the discussions in~\cite[][Chapter 6]{mollersteffensen} and~\cite{norberg1991}. The portfolio-wide means take the form
\begin{align*}
\bar{A}^g(t,s)
:=& \,
 \E\!\left[\left. A^g(t,s)\, \right| \mathcal{F}^S(t)\right]\!, \qquad t \leq s < \infty\\
\bar{V}^g(t) := &\, \E\!\left[\left. V^g(t)\, \right| \mathcal{F}^S(t)\right]\!, \\
\bar{V}^\star(t) :=&\, \E\!\left[\left. V^\star(t)\, \right| \mathcal{F}^S(t)\right]
\end{align*}
for $t\geq 0$. From~\eqref{eq:Vg} we find that $\bar{V}^g(t)$ may be obtained from $\bar{A}^g(t,\mathrm{d}s)$ via the equation
\begin{align}\label{eq:vgbar}
\bar{V}^g(t) = \int_t^n e^{-\int_t^s f(t,v)\md v} \bar{A}^g(t,\mathrm{d}s).
\end{align}
Thus it suffices to consider only $\bar{A}^g$ and $\bar{V}^\star$.

The portfolio-wide means represent values of liabilities under the assumption that the insurance portfolio is of such a size that unsystematic insurance risk can be disregarded. It corresponds to what is often referred to as mean-field approximations in the literature. In Subsection \ref{subsec:proj}, we show how to compute these.

We now turn our attention to the assets. They are described by a portfolio of $S$ which is self-financed by  the premium less benefits that the portfolio of insured pays to the insurer. We denote the value process by $U=\tuborg{U(t)}_{t\geq 0}$. We think of this process as the assets for the whole portfolio, but in our presentation the payments involved are only the contributions of a single insured. Since an individual insured pays $-\md B(t)$ to the insurer, this contribution to the total payments of the portfolio can be  represented by the expected cash flow $-\big(A^\circ(0,\mathrm{d}t)+A^b(0,\mathrm{d}t)\big)$. Thus we let $U$ take the form
\begin{align*}
\md U(t) = \theta(t)\md S_0(t) + \eta(t)\md S_1(t) - \big(A^\circ(0,\mathrm{d}t)+A^b(0,\mathrm{d}t)\big), \quad U(0) \equiv u_0, 
\end{align*}
where $\left(\theta, \eta\right) = \left(\theta(t), \eta(t) \right)_{t\geq 0}$ is a suitably regular $\mathcal{F}^S$-adapted investment strategy. We think of $\eta$ as a control variable for the insurer, since the number of units invested into the bank account is determined residually by $\theta(t) = (U(t)-\eta(t)S_1(t))/S_0(t)$. This gives
\begin{align}\label{eq:dU}
\mathrm{d}U(t)=r(t)\!\left(U(t)-\eta(t)S_1(t)\right)\!\!\md t + \eta(t)\md S_1(t) - \big(A^\circ(0,\mathrm{d}t)+A^b(0,\mathrm{d}t)\big).
\end{align}
In this paper, we only consider a single insured and the portfolio-wide mean reserves re\-present the contribution of this insured to the shape of the insurance business. To take into account multiple independent insured, one can consider $Z(0)$ as stochastic with distribution corresponding to the empirical distribution of initial states in the portfolio. The latter can be described by weights $w_j$ with the $j$'th weigth giving the proportion of insured that are initially in state $j\in\mathcal{J}$. The corresponding portfolio-wide means would in this case read
\begin{align*}
\sum_{j\in\mathcal{J}} w_j \E_j\!\left[\left. A^g(t,s)\, \right| \mathcal{F}^S(t)\right] \qquad \text{and} \qquad
\sum_{j\in\mathcal{J}} w_j \E_j\!\left[\left. V^\star(t)\, \right| \mathcal{F}^S(t)\right]\!,
\end{align*}
where $\E_j$ corresponds to expectation under the assumption that $Z(0) \equiv j$. Additionally, the insured typically belong to different cohorts implying that e.g.\ the transition rates and payment processes differ among insured. This is handled in a similar way. Also, the same considerations apply to the payments affecting the value process $U$. We consider these kinds of extensions from a single insured to a whole portfolio straightforward and do not give them further attention in the remainder of the paper.

Let $S(\cdot \wedge t)=\{S(u)\}_{0\leq u \leq t}$. We can now make the concepts of shape and controls precise.
\begin{definition}\label{def:state}
The \underline{shape of the insurance business} $\mathcal{I}$ is the triplet
\begin{align*}
\mathcal{I} = \left\{U(t),\bar{A}^g(t,\mathrm{d}s), \bar{V}^\star(t)\right\}_{t\geq 0},
\end{align*}
while the \underline{controls} are the pair $(\delta,\eta)=\left\{\delta(t), \eta(t)\right\}_{t\geq 0}$.
\end{definition}

\begin{assumption} \label{assump:controls}
We suppose that $(\delta,\eta)$ are chosen such that the setting is well-specified in the sense that $Q$ exists and is suitably regular. Furthermore, we assume that $\eta$ takes the form
\begin{align}\label{eq:eta}
\eta(t) = \eta\!\left(t,S(\cdot \wedge t), \mathcal{I}(t)\right)
\end{align}
for some explicitly computable and suitably regular deterministic mapping $\eta$, and we assume that $\delta$ takes the form
\begin{align} \label{eq:delta}
\begin{split}
\delta(t) = &\,\,\delta_0\left(t,S(\cdot \wedge t), Z(t), \mathcal{I}(t)\right) \\ &+
\delta_1\left(t,S(\cdot \wedge t), Z(t), \mathcal{I}(t)\right)\rho(\tau)^{\mathds{1}_{(\tau\leq t)}} 
\\ &+ \delta_2\left(t,S(\cdot \wedge t), Z(t), \mathcal{I}(t)\right) Q(t),
\end{split}
\end{align}
for some suitably regular deterministic mappings $\delta_0$, $\delta_1$ and $\delta_2$ that we are able to compute explicitly.
\end{assumption}

\begin{remark}\label{eq:Qneq0_delta}
In Remark~\ref{rmk:Qneq0} we discussed the extension to general $Q(0)$ and the idea of focusing on $Q(\cdot)-Q(0)$. By rewriting~\eqref{eq:delta} in the following manner,
\begin{align*}
\delta(t) = &\,\,\delta_0\left(t,S(\cdot \wedge t), Z(t), \mathcal{I}(t)\right) +  \delta_2\left(t,S(\cdot \wedge t), Z(t), \mathcal{I}(t)\right)Q(0)   \\ &+
\delta_1\left(t,S(\cdot \wedge t), Z(t), \mathcal{I}(t)\right)\rho(\tau)^{\mathds{1}_{(\tau\leq t)}} 
\\ &+ \delta_2\left(t,S(\cdot \wedge t), Z(t), \mathcal{I}(t)\right) (Q(t)-Q(0)),
\end{align*}
we see how this idea would manifest itself in relation to Assumption~\ref{assump:controls}. \demoo
\end{remark}

In the following, we also use the shorthand notations $t \mapsto \delta_i(t,Z(t))$, $i=0,1,2$, which only highlights $\mathcal{F}^Z$-measurable quantities.

The assumption that the controls depend only on portfolio-wide means rather than actual realizations of the balance sheet and the assets is the key choice of this paper. The risk we hereby account for is only the systematic risk, i.e.\ the risk that affects all insured.

Note that it is the assumption of $\delta$ being dependent on $U$ that makes $\eta$ a process that affects the payments to the insured, thus justifying it as a control. Note also that we allow $\delta$ to depend on $Z$, $\tau$, and $Q$, while this is not the case for $\eta$. This is since the dividends are allocated to the individual insured while the assets are a portfolio level quantity. The specific affine structure on $\delta$ mirrors that of $B$, cf.~\eqref{eq:B}. This is important for practical applications, as the following example highlights.

\begin{example}[Second order interest rate] \label{ex:rente}
The technical reserve $V^\star$ from \eqref{eq:technical} accumulates with the first order interest rate $r^\star$.\  Dividends may  then arise by accumulating the technical reserve with a second order interest rate $r^\delta$ that is  continuously readjusted based on the shape of the insurance business. This is obtained by letting   
\begin{align}\label{eq:delta_ex}
\delta(t) &= \big(r^\delta(t)-r^\star(t)\big)V^\star(t), \\ \label{eq:r_delta_phi}
r^\delta(t) &= \Phi(t,S(\cdot \wedge t), \mathcal{I}(t)), 
\end{align}
for some explicitly computable and suitably regular mapping $\Phi$.\ This corresponds to setting
\begin{align*}
\delta_0(t,j) &= \big(r^\delta(t)-r^\star(t)\big)1_{(j\in \Jp\setminus\tuborg{J})}\widetilde{V}_j^{\star}(t), \\
\delta_1(t,j) &= \big(r^\delta(t)-r^\star(t)\big)1_{(j\in \Jf\setminus\tuborg{2J+1})}\widetilde{V}_{j'}^{\star,+}(t) \\
\delta_2(t,j) &= \big(r^\delta(t)-r^\star(t)\big)V_j^{\star,\dagger}(t), 
\end{align*}
for all $j\in \J$. In case of a survival model, various specifications of second order interest rates is among the focal points of~\cite{mollersteffensen}. In discrete time, second order interest rates are also employed in~\cite{Bacinello2001}, where $r^\delta$ is taken to be the maximum of the technical interest rate $r^\star$ and a proportion (the so-called participation level) of the return on some reference portfolio. In the numerical example of Section~\ref{sec:num_ex}, we consider the following special case of $r^\delta$:  
\begin{align}\label{eq:r_delta_ex}
r^\delta(t) &= r^\star(t) + \frac{\kappa{\left(U(t)-\max\tuborg{\bar{V}^\star(t), \bar{V}^g(t)}\right)}^{\!+}}{\bar{V}^\star(t)}, \quad \kappa \in [0,1],
\end{align}      
where $\kappa$ is the share of some measure of excess assets used to buy additional benefits. This choice of $r^\delta$, where $r^\delta > r^\star$, in particular leads to a dividend design which also constitutes a minimum interest rate guarantee of $r^\star$. In case of a survival model, this type of $r^\delta$ also appears in e.g.~\cite[][Chapter 4.5.3]{mollersteffensen}.
\demo
\end{example}

The aim of this paper is to develop methods to compute the market value of bonus payments $V^b(0)$. Recall from Proposition~\ref{prop:FDB} that this can be done via the computation of the expected accumulated bonus cash flow $A^b(0,\cdot)$, which depends on the financial market through $Q$. To achieve this within the setup of Assumption~\ref{assump:controls}, we adopt a simulation approach. It follows from~\eqref{eq:Q} that for a simulated financial scenario, i.e.\ a realization of the whole path of $S$, we need the shape of the insurance business $\mathcal{I}(t) = (U(t),\bar{A}^g(t,\mathrm{d}s), \bar{V}^\star(t))$ and corresponding controls $(\delta(t),\eta(t))$ for all time points $t\geq 0$. In other words, starting today from time zero, we must \textit{project} the shape of the insurance business and the controls into future time points for each simulated financial scenario.

In the following sections, we formulate our \textit{scenario-based projection models} demonstrating how to project the shape of the insurance business in a specific financial scenario, and how to apply these projections to calculate the expected accumulated bonus cash flow $A^b(0,\cdot)$. Section~\ref{sec:projection} concerns the general case where $Q$ is allowed to be $\mathcal{F}^Z \vee \mathcal{F}^S$-adapted and where we apply \eqref{eq:bonusCF}--\eqref{eq:bonusCFa}. In the subsequent Section~\ref{sec:state_indep_projection} we specialize to $Q$ being state independent (of $Z$), i.e.\ $\mathcal{F}^S$-adapted, where we instead can apply the simpler formula~\eqref{eq:BonusCF_stateIndep}.

\section{Scenario-based projection model}\label{sec:projection} 
This section contains the main contributions of the paper and provides the foundation for the special case in Section~\ref{sec:state_indep_projection}. In Subsection~\ref{subsec:proj}, we formulate our general scenario-based projection model demonstrating how to project the shape of the insurance business into future time points in a given financial scenario. The projections are then in Subsection~\ref{subsec:q_mod} used to calculate the $Q$-modified transition probabilities $p_{z_0j}^Q(0,\cdot)$ and corresponding expected accumulated bonus cash flow $A^b(0,\cdot)$. Based on this, we present in Subsection~\ref{subsec:proj_implement} a procedure for the computation of $V^b(0)$ via an application of Proposition~\ref{prop:FDB}.

As noted in Proposition \ref{prop:FDB}, we are able to simplify calculations of $A^b(0,\cdot)$ to what we coin \textit{state-independent} calculations of $Q$ and $p$ if $Q$ is assumed $\mathcal{F}^S$-adapted. This special case leads to a notion of a \textit{state-independent scenario-based projection model}, which is studied in more details in Section~\ref{sec:state_indep_projection}.

\subsection{Projecting the shape}\label{subsec:proj}

We now turn our attention to projection of the shape of the insurance business. This consists of computation of $\mathcal{I}=(U,\bar{A}^g,\bar{V}^\star)$ for realizations of $S$, where each realization exactly represents a simulated financial scenario.

The method for computation of $U$ for a realization of $S$ follows immediately from the dynamics of the assets according to~\eqref{eq:dU}. The computational issue reduces to that of computing $p^Q_{z_0j}(0,\cdot)$, cf.~\eqref{eq:bonusCF}--\eqref{eq:bonusCFa} and~\eqref{eq:dU}. Thus we focus on the projection of the portfolio-wide means $\bar{A}^g$ and $\bar{V}^\star$.

First, we consider the portfolio-wide means $\bar{A}^g$ of the expected accumulated guaranteed cash flows $A^g$.
\begin{proposition}\label{prop:projmarket}
The portfolio-wide means $\bar{A}^g$ of the expected accumulated guaranteed cash flows $A^g$ read
\begin{align*}
\bar{A}^g(t,\mathrm{d}s)=A^\circ(0,\mathrm{d}s)+\sum_{j\in \J}p^Q_{z_0j}(0,t)A^\dagger_j(t,\mathrm{d}s)
\end{align*}
for all $t \geq 0$. 
\end{proposition}
\begin{proof}
By~\eqref{eq:CF}, \eqref{def:pQ}, and due to the assumed independence between $Z$ and $S$, we immediately find that
\begin{align*}
\bar{A}^g(t,s)  
&=
\E\!\left[\left. A^\circ(t,s) \, \right| \mathcal{F}^S(t)\right]  + \sum_{j \in \J}  \E\!\left[\left. \mathds{1}_{(Z(t) = j)} Q(t)A^{\dagger}_{Z(t)}(t, s) \, \right| \mathcal{F}^S(t)\right]  \\
&=
\E\!\left[A^\circ(t,s) \right] + \sum_{j \in \J} p^Q_{z_0j}(0,t)A^\dagger_j(t,s).
\end{align*}
By~\eqref{eq:unit_cf_def} and the iterated law of expectations,
\begin{align*}
\E\!\left[A^\circ(t,s) \right]
&=
\E\!\left[B^\circ(s) - B^\circ(t) \right] \\
&=
A^\circ(0,s) - \E\!\left[B^\circ(t) - B^\circ(0) \right]\!.
\end{align*}
Since the latter term does not depend on $s$, we find that
\begin{align*}
\bar{A}^g(t,\md s)
=
A^\circ(0,\md s) + \sum_{j \in \J} p^Q_{z_0j}(0,t)A^\dagger_j(t,\md s)
\end{align*}
as desired.
\end{proof}
Consequently, given $A^\circ$ and $A^\dagger$ the computational issue has been reduced to that of computing the $Q$-modified transition probabilities $p^Q_{z_0j}(0,\cdot)$.

Next we consider the portfolio-wide mean of the technical reserve of guaranteed payments, $\bar{V}^\star$. We could follow the same approach above and calculate the technical reserves via expected (accumulated) cash flows, however, since the technical interest rate is deterministic, a range of technical reserves, including $V^{\star,\dagger}$, $\widetilde{V}^{\star}$, and $\widetilde{V}^{\star,+}$, can be computed more efficiently by solving the differential equations of Thiele type derived from~\eqref{eq:thiele_tech}, cf.~Subsection~\ref{subsec:unit_val} and Subsection~\ref{subsec:div_bon}.

Denote by $\bar{V}^{\star,\circ}$ the portfolio-wide mean technical reserves of predetermined payments given by
\begin{align*}
\bar{V}^{\star,\circ}(t) = \E\!\left[\left. V^{\star,\circ}(t)\, \right| \mathcal{F}^S(t)\right]
\end{align*}
for $t\geq 0$. Since $Z$ and $S$ are assumed independent, we could replace the conditional expectation by an ordinary expectation.

\begin{proposition}\label{prop:projtechnical}
The portfolio-wide mean technical reserve of guaranteed payments reads
\begin{align*}
\bar{V}^\star(t) &= \bar{V}^{\star,\circ}(t) + \sum_{j\in \J}p^Q_{z_0j}(0,t)V^{\star,\dagger}_{j}(t),
\end{align*}
while the portfolio-wide mean technical reserve of predetermined payments reads
\begin{align}\label{eq:projtechnicalunit}
\bar{V}^{\star,\circ}(t) &=\sum_{j\in \Jp\atop j\neq J}p_{z_0j}(0,t)\widetilde{V}_j^{\star}(t)+\sum_{j\in \Jf \atop j\neq 2J+1}p^\rho_{z_0j}(0,t)\widetilde{V}_{j'}^{\star,+}(t).
\end{align}
\end{proposition}
\begin{proof}
By~\eqref{eq:technical} and \eqref{def:pQ}, direct calculations yield
\begin{align*}
\bar{V}^\star(t)
&= \E\!\left[\left. V^{\star,\circ}(t) \, \right| \mathcal{F}^S(t) \right] + \sum_{j\in \J}\E\!\left[\left. \mathds{1}_{(Z(t) = j)} Q(t) V_{Z(t)}^{\star,\dagger}(t)\, \right| \mathcal{F}^S(t) \right] \\
&=  \bar{V}^{\star, \circ}(t) + \sum_{j\in \J}p_{z_0j}^Q(0,t) V_j^{\star,\dagger}(t).
\end{align*}
To obtain~\eqref{eq:projtechnicalunit}, we split $V^{\star,\circ}$ according to the events of $Z(t)$ being in $\Jp\setminus\tuborg{J}$, $\Jf\setminus\tuborg{2J+1}$, and $\{J,2J+1\}$. According to~\eqref{eq:TechnicalUnit}, we then have
\begin{align*}
\bar{V}^{\star,\circ}(t) &=  \E\!\left[\left. \mathds{1}_{\left(Z(t)\in \Jp\setminus\tuborg{J}\right)}\widetilde{V}_{Z(t)}^{\star}(t) + \mathds{1}_{\left(Z(t)\in \Jf\setminus\tuborg{2J+1}\right)}\rho(\tau)\widetilde{V}_{Z(t)'}^{\star,+}(t) \, \right| \mathcal{F}^S(t)\right] \\
&=\E\!\Bigg[ \sum_{j\in \Jp\atop j\neq J} \mathds{1}_{(Z(t)=j)}\widetilde{V}_{j}^{\star}(t) +  \sum_{j \in \Jf\atop j\neq 2J+1} \mathds{1}_{(Z(t)=j)}\rho(\tau)\widetilde{V}_{j'}^{\star,+}(t) \, \Bigg| \mathcal{F}^S(t)\Bigg] \\
&= \sum_{j\in \Jp\atop j\neq J}p_{z_0j}(0,t)\widetilde{V}_j^{\star}(t)+\sum_{j\in \Jf\atop j\neq 2J+1}p^\rho_{z_0j}(0,t)\widetilde{V}_{j'}^{\star,+}(t),
\end{align*}
as desired.
\end{proof}

As already mentioned, the technical reserves $V^{\star,\dagger}$, $\widetilde{V}^{\star}$, and $\widetilde{V}^{\star,+}$ can be computed efficiently using differential equations of Thiele type, while the $\rho$-modified transition probabilities are simply computed according to~\eqref{eq:kol_mod}. Thus Proposition~\ref{prop:projtechnical} reduces the computational complexity to that of computing $Q$-modified transition probabilities $p_{z_0j}^Q(0,\cdot)$. This computation is studied in details in the next subsection.

\subsection{$Q$-modified transition probabilities}\label{subsec:q_mod}

We are now ready to present a system of differential equations for the $Q$-modified transition probabilities $p_{z_0j}^Q(0,\cdot)$; here $p^\rho_{z_0j}(0,\cdot) := p_{z_0j}(0,\cdot)$ for $z_0 \in \Jf$, which is in accordance with $\tau = 0$ for $z_0 \in \Jf$ and the assumption $\rho(0) = 1$. 

\begin{theorem}\label{thm:diffQ}
The $Q$-modified transition probabilities $p^Q_{z_0j}(0,\cdot)$ satisfy for $j\in\J$ the differential equations
\begin{align}\label{eq:diff_pQ}
\begin{split}
\deriv{}{t}p^Q_{z_0j}(0,t)
= \, &
\frac{
p_{z_0j}(0,t)\delta_0(t,j) + p^\rho_{z_0j}(0,t)\delta_1(t,j) + p^Q_{z_0j}(0,t)\delta_2(t,j) 
}
{
V_j^{\star,\dagger}(t)
} \\
&-
p^Q_{z_0j}(0,t)\mu_{j\bullet}(t)
+\sum_{k\in \J\atop k\neq j}p^Q_{z_0k}(0,t)\mu_{kj}(t),
\hspace{10mm}  p^Q_{z_0j}(0,0)=0.
\end{split}
\end{align}
\end{theorem}
\begin{proof}
The boundary conditions follows by the assumption that $Q(0)=0$. Referring to~\eqref{def:pQ} and \eqref{eq:Q}, we have
\begin{align*}
p^Q_{z_0j}(0,t)
&=
\E\!\left[\left. \mathds{1}_{(Z(t)=j)}Q(t) \, \right|  \mathcal{F}^{S}(t)\right] = \E\!\left[\left. \mathds{1}_{(Z(t)=j)} \int_0^t \frac{\delta(u)}{V^{\star,\dagger}_{Z(u)}(u)} \, \md u  \, \right|  \mathcal{F}^{S}(t)\right] 
\end{align*}
with 
\begin{align*}
\delta(t) = \delta_0(t,Z(t)) + \delta_1(t,Z(t))\rho(\tau)^{\mathds{1}_{(\tau\leq t)}} +\delta_2(t,Z(t))Q(t).
\end{align*}
Note that for $0\leq u \leq t$ and $k\in\J$,
\begin{align*}
\E\!\left[\left. \mathds{1}_{(Z(u)=k)} \frac{p^Q_{z_0k}(0,u)}{p_{z_0k}(0,u)} \, \right|  \mathcal{F}^{S}(t)\right]
&=
\E\!\left[\left. \mathds{1}_{(Z(u)=k)}Q(u) \, \right|  \mathcal{F}^{S}(t)\right]\!, \\
\E\!\left[ \mathds{1}_{(Z(u)=k)} \frac{p^\rho_{z_0k}(0,u)}{p_{z_0k}(0,u)}  \right]
&=
\E\!\left[ \mathds{1}_{(Z(u)=k)}\rho(\tau)^{\mathds{1}_{(\tau\leq u)}}  \right]\!.
\end{align*}
Thus by Markovianity of $Z$ and independence between $Z$ and $S$,
\begin{align} \label{eq:retro_rep}
p^Q_{z_0j}(0,t)
&=
\E\!\left[\left. \mathds{1}_{(Z(t) = j)} \int_0^t \sum_{k \in \J} \mathds{1}_{(Z(u)=k)} b_k^Q(u) \, \md u  \, \right|  \mathcal{F}^{S}(t)\right]
\end{align}
with $b_k^Q$, $k\in\J$, given by
\begin{align}\label{eq:retro_pay}
b_k^Q(u)
=
\frac{
\delta_0(u,k) + \delta_1(u,k)\frac{
p^\rho_{z_0k}(0,u)
}
{
p_{z_0k}(0,u)
}
+
\delta_2(u,k)\frac{
p^Q_{z_0k}(0,u)
}
{
p_{z_0k}(0,u)
}
}
{
V^{\star,\dagger}_{k}(u)
}
\end{align}
for all $u \geq 0$. The assumption of independence between $Z$ and $S$, Markovianity of $Z$, and Fubini's theorem finally yield
\begin{align}\label{eq:pq_rep}
p^Q_{z_0j}(0,t)
=
\int_0^t \sum_{k \in \J} p_{z_0k}(0,u) p_{kj}(u,t) b_k^Q(u) \, \md u.
\end{align}
The statement of the theorem is now established by differentiation as follows. Leibniz' integration rule gives
\begin{align*}
\deriv{}{t} p^Q_{z_0j}(0,t) &= \sum_{k\in \J}1_{(k=j)}p_{z_0k}(0,t)b_k^Q(t)+\int_0^t \sum_{k\in \J} p_{z_0k}(0,u)\left(\deriv{}{t}p_{kj}(u,t)\right)  b_k^Q(u) \md u \\
&=
\frac{
\delta_0(t,j)p_{z_0j}(0,t) + \delta_1(t,j)p^\rho_{z_0j}(0,t)+\delta_2(t,j)p^Q_{z_0j}(0,t)
}
{
V^{\star,\dagger}_{j}(t)
} \\
& \hspace{4mm} +
\int_0^t \sum_{k\in \J} p_{z_0k}(0,u)\left(\deriv{}{t}p_{kj}(u,t)\right)  b_k^Q(u) \md u.
\end{align*} 
Applying Kolmogorov's forward differential equations and~\eqref{eq:pq_rep} to the last line of the equation we find that
\begin{align*}
\deriv{}{t} p^Q_{z_0j}(0,t) = \, &
\frac{
\delta_0(t,j)p_{z_0j}(0,t) + \delta_1(t,j)p^\rho_{z_0j}(0,t)+\delta_2(t,j)p^Q_{z_0j}(0,t)
}
{
V^{\star,\dagger}_{j}(t)
} \\
& - p_{z_0j}^Q(0,t)\mu_{j\bullet}(t)
+ \sum_{\ell \in \J\atop \ell \neq j} p_{z_0\ell}^Q(0,t) \mu_{\ell j}(t)
\end{align*}
as desired. 
\end{proof}

\begin{remark}
There exists a clear link between $Q$-modified transition probabilities and so-called state-wise retrospective reserves. Referring to~\eqref{eq:retro_rep} and~\eqref{eq:retro_pay}, we see that for a fixed financial scenario,
\begin{align*}
W_j(\cdot) :=
\frac{
p_{z_0j}^Q(0,\cdot)
}
{
p_{z_0j}(0,\cdot)
}
\end{align*}
corresponds to the state-wise retrospective reserve of~\cite{norberg1991} (in the presence of information $\mathcal{G}(t) = \mathcal{F}^S(t) \vee \sigma(Z(t))$, cf.~\cite{norberg1991} Subsection 5.B) with payments
\begin{align*}
-\sum_{j \in \J} \mathds{1}_{(Z(t) = j)} b^Q_j(t) \, \md t
\end{align*}
and interest rate zero. Contrary to the primary setup of~\cite{norberg1991}, the payments considered here are functions of the state-wise retrospective reserves $W_j(\cdot)$. \demoo
\end{remark}  

The system of differential equations for $p_{z_0j}^Q(0,\cdot)$ from Theorem~\ref{thm:diffQ} involves the shape of the insurance business $\mathcal{I}$ through the mappings $\delta_0$, $\delta_1$, and $\delta_2$.
Together with the results of the previous subsection, Theorem~\ref{thm:diffQ} allows us to formulate a procedure for the calculation of $V^b(0)$. The procedure is presented in the next subsection.

\subsection{Numerical procedure}\label{subsec:proj_implement}
Based on the results of the previous subsections, we demonstrate a procedure for the scenario-based projection model. An actual numerical example is given later in Section~\ref{sec:num_ex}. In what follows, we suppose we are given mappings $(\delta,\eta)$ serving as controls. They are assumed to satisfy Assumption~\ref{assump:controls}.

Besides the financial scenarios, the input consists of the following quantities which can be precalculated independently of the financial scenarios:
\begin{enumerate}[(1)]
\item The expected accumulated cash flow of predetermined payments $ A^\circ(0,s)$ for $s\geq0$ as in~\eqref{eq:UnitCFjp}.   
\item The portfolio-wide mean technical reserve of predetermined payments $\bar{V}^{\star,\circ}(t)$ for all $t\geq 0$ calculated via~\eqref{eq:projtechnicalunit}. 
\item For each $t\geq 0$, state-wise expected accumulated unit bonus cash flows $A_j^\dagger(t,s)$ for all $s\geq t$ and $j\in \J$ as in~\eqref{eq:CFtilde_state}--\eqref{eq:CFtilde_state_a}. 
\item State-wise technical unit reserves $V^{\star,\dagger}_j(t)$ for all $t\geq 0$ and $j\in \J$~as in \eqref{eq:TechnicalTilde}.   
\item Transition probabilities $p_{z_0j}(0,t)$ for all $t\geq 0$ and $j\in \J$.
\end{enumerate}
As discussed previously, this input can be calculated using classic methods for solving differential equations of Thiele type as well as ($\rho$-modified) Kolmogorov forward dif\-ferential equations.

The financial scenarios are $N$ realizations $\{S^k(t)\}_{t\geq0}$, $k=1,\ldots,N$, of $\{S(t)\}_{t\geq0}$ with corresponding short rate $r^k$ and forward rate curves $f^k$. We consider them as output of an economic scenario generator.

The procedure essentially consists of computing $p^{Q}_{z_0j}(0,\cdot)$, $j\in\J$, and $U(\cdot)$ in each financial scenario by solving a system of (stochastic) differential equations. The involved part is to evaluate the differentials. The procedure looks as follows. For each financial scenario $k=1,\ldots,N$:
\begin{itemize}
\item Initialize with $p_{z_0j}^{Q,k}(0,0) = 0$ for all $j\in \J$ and $U^k(0)=u_0$. 
\item Apply a numerical algorithm to solve the coupled (stochastic) differential equation systems for $p^{Q,k}_{z_0j}(0,\cdot)$, $j\in\J$, and $U^k(\cdot)$ from Theorem~\ref{thm:diffQ} and~\eqref{eq:dU}, respectively.
\begin{itemize}
\item Evaluating the differentials at time $t$ involves the mappings $(\delta_0,\delta_1,\delta_2,\eta)$ from~\eqref{eq:eta}--\eqref{eq:delta}. By inspection of the differentials and these mappings, we see that we require the shape of the insurance business
\begin{align*}
\mathcal{I}^k(t)=\left(U^k(t),\bar{A}^{g,k}(t,\mathrm{d}s),\bar{V}^{\star,k}(t)\right)\!,
\end{align*}
the expected bonus cash flow $a^{b,k}(0,t)$, as well as the input. Computation of $\bar{A}^{g,k}(t,\mathrm{d}s)$, $\bar{V}^{\star,k}(t)$, and $a^{b,k}(0,t)$ is achieved via Proposition~\ref{prop:projmarket}, Proposition~\ref{prop:projtechnical}, and~\eqref{eq:bonusCFa}.
\end{itemize}  
\item We emphasize that as part of evaluating the differentials we computed the expected bonus cash flow $a^{b,k}(0,\cdot)$.
\end{itemize}
The procedure completes by computing the market value of bonus payments $V^b(0)$ via
\begin{align*}
 V^b(0) &\approx
\frac{1}{N}\sum_{k=1}^N \int_0^n e^{-\int_0^{t} r^k(v) \md v} a^{b,k}(0,t)\md t
\end{align*} 
using an algorithm for numerical integration.

Note that we require the input (3), which are the state-wise expected accumulated unit bonus cash flows $A_j^\dagger(\cdot,\cdot)$ evaluated on the two-dimensional time grid $\{(t,s) \in [0,\infty)^2 : t \leq s\}$. To precompute this input, one must solve Kolmogorov's forward differential equations many times, once for every $t \geq 0$ and $j\in\J$. This significantly impacts the numerical efficiency of the procedure. Furthermore, the algorithm itself depends on the market basis for the specific insured through the transition rates $\mu$. In practice, where the algorithm must be executed for many insured, one must view the specific transition rates for a single insured as input.

In general, it is preferable to use analytical methods and numerical methods for ordinary differential equations compared to Monte Carlo methods. For example, solving Kolmogorov's forward differential equations in order to calculate an expected cash flow may be done orders of magnitudes faster (for a given precision requirement) compared to calculating the same expected cash flow via Monte Carlo methods. The procedure presented above shows how to disentangle biometric and behavioral risk from financial risk, allowing us to solve part of the problem via numerical methods for differential equations instead of using Monte Carlo methods. For a given precision requirement, this significantly reduces the time required to handle biometric and behavioral risk. If the portfolio merely consists of a single insured, which we presumed, it is however the simulation of financial risk that constitutes the main numerical complexity and time usage. As the number of insured increases, so does the relative time needed to calculate expected cash flows, since the same financial scenarios are used across all insured, and thus the numerical speed up of our procedure, compared with a full-blown Monte Carlo approach, should become significant.

In the following section, we present the simpler state-independent scenario-based projection model, where we require that the dividend strategy be specified (or approximated) such that $Q$ is $\mathcal{F}^S$-adapted. By presenting a numerical procedure for the model, we show how this requirement on the dividend strategies leads to a numerical speedup.

\section{State-independent scenario-based projection model}\label{sec:state_indep_projection}

This section concerns the formulation of the state-independent scenario-based projection model. The model is a special case of the projection model from Section~\ref{sec:projection} which relies on ensuring $Q$ to be an $\mathcal{F}^{S}$-adapted process such that the simplified case of Proposition~\ref{prop:FDB} applies.

In Subsection~\ref{subsec:class_div}, we provide sufficient conditions on $\delta$ such that $Q$ is $\mathcal{F}^{S}$-adapted. Next, Subsection~\ref{subsec:pro_shape_re} revisits the projection of the shape under this simplification. Finally, in Subsection~\ref{subsec:implement_indep} we present a procedure for the computation of the market value of bonus payments in the state-independent projection model.

\subsection{Background}\label{subsec:background}

The concept of state-independent modeling is uncommon in the literature, with the projection model described in~\cite[][Section 4]{JensenSchomacker2015} being one of few exceptions. It is our impression that projections models such as the one found in~\cite[][Section 4]{JensenSchomacker2015} have been implemented or are being implemented in practice, which further underlines the importance of studying state-independent scenario-based projection models in more detail.

In the projection model described in~\cite[][Section 4]{JensenSchomacker2015}, additional benefits are bought according to the portfolio-wide mean $\bar{V}^{\star,\dagger}$ of the technical reserve rather than the actual technical reserve $V^{\star,\dagger}_{Z(\cdot)}$, see~\cite[][p.\ 196]{JensenSchomacker2015}. Furthermore, the dividend yield is implicitly assumed $\mathcal{F}^S$-adapted, since in~\cite{JensenSchomacker2015} the evolution of the policy is only described on an averaged portfolio level. In unity, this leads to an $\mathcal{F}^S$-adapted $Q$.  In the following, we provide sufficient conditions on $\delta$ within our setup to ensure this adaptability. Later, in Example~\ref{ex:second_order_int} in Subsection~\ref{subsec:pro_shape_re}, we provide a slightly more explicit link to the projection model in~\cite[][Section 4]{JensenSchomacker2015}.

\subsection{Class of dividend strategies}\label{subsec:class_div}

Recall from~\eqref{eq:Q} and~\eqref{eq:delta} that $Q$ is the solution to the differential/integral equation
\begin{align*}
\md Q(t) = \frac{\delta_0(t,Z(t)) + \delta_1(t,Z(t))\rho(\tau)^{\mathds{1}_{(\tau\leq t)}}+\delta_2(t,Z(t))Q(t)}{V^{\star,\dagger}_{Z(t)}(t)}\md t, \quad Q(0)=0.
\end{align*}
To ensure that $Q$ is an $\mathcal{F}^{S}$-adapted process, it suffices to require that $\delta_0$, $\delta_1$ and $\delta_2$ are on the form  
\begin{align}\label{eq:deltaform}
\delta_i(t,Z(t)) &= \widetilde{\delta}_i(t) V_{Z(t)}^{\star,\dagger}(t), \quad \quad i=0,2, \\ \label{eq:deltaform2}
\delta_1(t,Z(t)) &= 0, 
\end{align}
where we have used the shorthand notation $\widetilde{\delta}_i(t) = \widetilde{\delta}_i\left(t,S(\cdot \wedge t), \mathcal{I}(t)\right)$ for suitably regular deterministic mappings $\widetilde{\delta}_i$, $i=0,2$. This is a consequence of the following observation. When~\eqref{eq:deltaform}--\eqref{eq:deltaform2} hold, then simply
\begin{align}\label{eq:Qtilde}
\md Q(t) = \big(\, \widetilde{\delta}_0(t) + \widetilde{\delta}_2(t)Q(t)\big)\!\md t, \quad Q(0) = 0.
\end{align}
This implies $p^Q_{z_0j}(0,t) = Q(t) p_{z_0j}(0,t)$, cf.~\eqref{eq:pQ_split}.

\begin{remark}\label{rmk:non_affine_stateindep}
Since the class of dividend strategies presented here builds on Assumption~\ref{assump:controls}, affinity in $Q$ is more or less implicitly assumed. The simplifications we obtain in the following Subsections~\ref{subsec:pro_shape_re}--\ref{subsec:implement_indep} build on $Q$ being $\mathcal{F}^S$-adapted rather than the dividend strategy being affine in $Q$. The results are therefore trivially extendable to dividend strategies that are non-affine in the number of additional benefits held. \demoo
\end{remark}

\subsection{Projecting the shape revisited}\label{subsec:pro_shape_re}

For the portfolio-wide means $\bar{A}^g$ we observe a simplification in the part that concerns future bonus payments similar to what we previously saw concerning the predetermined payments:

\begin{corollary}\label{col:state_indep_projmarket}
Assume that the dividend strategy $\delta$ is on the form~\eqref{eq:deltaform}--\eqref{eq:deltaform2}. The portfolio-wide means $\bar{A}^g$ of the expected accumulated guaranteed cash flows $A^g$ then read
\begin{align*}
\bar{A}^g(t,\mathrm{ds})=A^\circ(0,\mathrm{d}s)+Q(t)A^\dagger(0,\mathrm{d}s).
\end{align*}
\end{corollary}
\begin{proof}
From Proposition~\ref{prop:projmarket} and its proof, we have
\begin{align*}
\bar{A}^g(t,s) = A^\circ(0,s) - E\!\left[B^\circ(t) - B^\circ(0)\right] + \E\!\left[\left. Q(t)A^{\dagger}(t, s) \, \right| \mathcal{F}^S(t)\right].
\end{align*}
Since by assumption $Q$ is $\mathcal{F}^S$-adapted and $Z$ and $S$ are independent, referring to~\eqref{eq:Vcirc} with superscript $\circ$ replaced by $\dagger$ and applying the law of iterated expectations yields
\begin{align*}
\E\!\left[\left. Q(t)A^{\dagger}(t, s) \, \right| \mathcal{F}^S(t)\right]
&=
Q(t) \E\!\left[B^{\dagger}(s) - B^{\dagger}(t)\right] \\
&=
Q(t)A^{\dagger}(0,s) - Q(t) \E\!\left[B^{\dagger}(t) - B^{\dagger}(0)\right]
\end{align*}
Consequently,
\begin{align*}
\bar{A}^g(t, \md s)
=
A^\circ(0,\md s) + Q(t) A^\dagger(0, \md s) 
\end{align*}
as desired.
\end{proof}

For the technical reserve, the result is similar. Before we present the result, let the portfolio-wide mean technical unit bonus reserve $\bar{V}^{\star,\dagger}$ be given by
\begin{align*}
\bar{V}^{\star,\dagger}(t) = \E\!\left[\left. V_{Z(t)}^{\star,\dagger}(t)\, \right| \mathcal{F}^S(t)\right]
\end{align*}
for $t\geq 0$. Since $Z$ and $S$ are assumed independent, we could replace the conditional expectation by an ordinary expectation. It is then a trivial observation that
\begin{align}\label{eq:projtechnicalunitDagger}
\bar{V}^{\star,\dagger}(t) = \sum_{j\in\J} p_{z_0j}(0,t) V_j^{\star,\dagger}(t).
\end{align}
\begin{corollary}\label{col:state_indep_projtechnical}
Assume that the dividend strategy $\delta$ is on the form~\eqref{eq:deltaform}--\eqref{eq:deltaform2}. The portfolio-wide mean technical reserve of guaranteed payments then reads
\begin{align}\nonumber
\bar{V}^\star(t) &= \bar{V}^{\star,\circ}(t) + Q(t)\bar{V}^{\star,\dagger}(t).
\end{align}
\end{corollary}
\begin{proof}
Since by assumption, $Q$ is $\mathcal{F}^S$-adapted and $Z$ and $S$ are independent, the result follows immediately from~\eqref{eq:pQ_split}, Proposition~\ref{prop:projtechnical}, and~\eqref{eq:projtechnicalunitDagger}.
\end{proof}
The following example is a continuation of Example~\ref{ex:rente} regarding the accumulation of the technical reserve with a second order interest rate.
\begin{example}[Second order interest rate continued] \label{ex:second_order_int}
The dividend strategy from Example \ref{ex:rente} regarding accumulation of the technical reserve $V^\star$ with a second order interest rate $r^\delta$ does not satisfy the requirements on $\delta$ from \eqref{eq:deltaform}--\eqref{eq:deltaform2}. Instead, the strategy 
\begin{align}\label{eq:rente2}
\delta(t) &= \big(r^\delta(t)-r^\star(t)\big)\frac{\bar{V}^{\star}(t)}{\bar{V}^{\star,\dagger}(t)} V^{\star,\dagger}_{Z(t)}(t),
\end{align}
satisfies \eqref{eq:deltaform}--\eqref{eq:deltaform2} with
$$\widetilde{\delta}_0(t) =  (r^\delta(t)-r^\star(t))\frac{\bar{V}^{\star,\circ}(t)}{\bar{V}^{\star,\dagger}(t)} \quad  \text{and} \quad 
\widetilde{\delta}_2(t)=(r^\delta(t)-r^\star(t)).$$
One may think of this strategy as an accumulation of the portfolio-wide mean technical reserve $\bar{V}^\star$ with $r^\delta$ instead, since by \eqref{eq:Qtilde},
\begin{align*}
\bar{V}^{\star,\dagger}(t)\md Q(t) = \big(r^\delta(t)-r^\star(t)\big)\bar{V}^{\star}(t)\md t.
\end{align*}
This is in accordance with the projection model in~\cite[][Section 4]{JensenSchomacker2015}. By multiplying the strategy \eqref{eq:rente2} with 
\begin{align*}
\frac{V^\star(t)}{\bar{V}^{\star}(t)} \quad \text{and} \quad \frac{\bar{V}^{\star,\dagger}(t)}{V^{\star,\dagger}_{Z(t)}(t)}
\end{align*}
one arrives at strategy of Example~\ref{ex:rente}. If the two ratios are close to one, the strategy~\eqref{eq:rente2} approximates the strategy of Example~\ref{ex:rente}. Note that $\E\!\left[\left. V^\star(t) / \, \bar{V}^{\star}(t) \, \right| \mathcal{F}^S(t)\right]=1$, i.e.\ the portfolio-wide mean of the first ratio is equal to one. For the latter ratio this is not necessarily the case, since it is non-linear in $V^{\star,\dagger}_{Z(t)}(t)$. 

Even if the ratioes are not close to one, the strategy given by~\eqref{eq:rente2} may still approximate the strategy of Example~\ref{ex:rente} in terms of producing an akin time zero market value of bonus payments. Taking as starting point the strategy of Example~\ref{ex:rente} and the expression of Remark~\ref{rmk:approx_expression}, we find that
\begin{align*}
V^b(0) = \int_0^n \E\!\left[e^{-\int_0^s r(v)\md v} \big(r^\delta(s)-r^\star(s)\big) V^\star(s)\frac{V^\dagger_{Z(s)}(s)}{V_{Z(s)}^{\star,\dagger}(s)}\right] \mathrm{d}s,
\end{align*}
so that by the law of iterated expectations,
\begin{align}\label{eq:vb0exSD}
V^b(0) = \int_0^n \E\!\left[e^{-\int_0^s r(v)\md v} \big(r^\delta(s)-r^\star(s)\big) \E\!\left[\left.  V^\star(s)\frac{V^\dagger_{Z(s)}(s)}{V_{Z(s)}^{\star,\dagger}(s)} \, \right| \mathcal{F}^S(s) \right]\right] \mathrm{d}s.
\end{align}
Taking instead as a starting point the strategy given by~\eqref{eq:rente2}, we find the expression
\begin{align}\label{eq:vb0exSI}
\int_0^n \E\!\left[e^{-\int_0^s r(v)\md v} \big(r^\delta(s)-r^\star(s)\big) \bar{V}^{\star}(s) \frac{V_{Z(s)}^\dagger(s)}{\bar{V}^{\star,\dagger}(s)}\right]\mathrm{d}s
\end{align}
for the corresponding time zero market value of bonus payments. We should like to hightlight that the second order interest rate $r^\delta$ depends on $Q$ through the shape of the insurance business, so that the second order interest rates of~\eqref{eq:vb0exSD} and~\eqref{eq:vb0exSI} may differ.

From~\eqref{eq:vb0exSD} and~\eqref{eq:vb0exSI}, we see that the strategy given by~\eqref{eq:rente2} in particular leads to a decent approximation of the market value of bonus payments as long as
\begin{align*}
\big(r^\delta(s)-r^\star(s)\big)
\E\!\left[\left.  V^\star(s)\frac{V^\dagger_{Z(s)}(s)}{V_{Z(s)}^{\star,\dagger}(s)} \, \right| \mathcal{F}^S(s) \right]
\approx
\big(r^\delta(s)-r^\star(s)\big)
\bar{V}^\star(s) \frac{\E\!\left[\left. V^\dagger_{Z(s)}(s) \, \right| \mathcal{F}^S(s)\right]}{\bar{V}^{\star,\dagger}(s)}
\end{align*}
for all $s\geq0$. This is for example the case if the second order interest rates only differ ever so slightly and
\begin{align*}
\E\!\left[\left.  V^\star(s)\frac{V^\dagger_{Z(s)}(s)}{V_{Z(s)}^{\star,\dagger}(s)} \, \right| \mathcal{F}^S(s) \right]
\approx
\bar{V}^\star(s) \frac{\E\!\left[\left. V^\dagger_{Z(s)}(s) \, \right| \mathcal{F}^S(s)\right]}{\bar{V}^{\star,\dagger}(s)}
\end{align*}
for all $s\geq0$. The latter is by definition of the portfolio-wide mean technical reserve $\bar{V}^\star$ indeed the case if the safety loading $s \mapsto V^\dagger_{Z(s)}(s)/V_{Z(s)}^{\star,\dagger}(s)$ is approximately state-independent, i.e.\ does not depend significantly on $Z$ (whenever $\bar{V}^\star$ is non-zero).\demo
\end{example}

\subsection{Numerical procedure}\label{subsec:implement_indep}

Based on the results of the previous subsections, we demonstrate a procedure for the state-independent scenario-based projection model. An actual numerical example is given later in Section~\ref{sec:num_ex}. In what follows, we suppose we are given mappings $(\delta,\eta)$ serving as controls. They are assumed to satisfy Assumption~\ref{assump:controls} with $\delta$ on the form~\eqref{eq:deltaform}--\eqref{eq:deltaform2}. 

Besides the financial scenarios, the input consists of the following quantities which can be precalculated independently of the financial scenarios:
\begin{enumerate}[(1)]
\item The expected accumulated cash flow of predetermined payments $ A^\circ(0,s)$ for all $s\geq 0$ as in~\eqref{eq:UnitCFjp}.   
\item The portfolio-wide mean technical reserve of predetermined payments $\bar{V}^{\star,\circ}(t)$ for all $t\geq 0$ calculated via~\eqref{eq:projtechnicalunit}. 
\item The expected unit bonus cash flow $ a^\dagger(0,s)$ for all $s \geq 0$ as in~\eqref{eq:CFtilde_a}. 
\item The portfolio-wide mean technical unit bonus reserve $\bar{V}^{\star,\dagger}(t)$ for all $t\geq 0$ calculated via~\eqref{eq:projtechnicalunitDagger}
\end{enumerate}
As discussed previously, this input can be calculated using classic methods for solving differential equations of Thiele type as well as ($\rho$-modified) Kolmogorov forward dif\-ferential equations.

The financial scenarios are $N$ realizations $\{S^k(t)\}_{t\geq0}$, $k=1,\ldots,N$, of $\{S(t)\}_{t\geq0}$ with corresponding short rate $r^k$ and forward rate curves $f^k$. We consider them as output of an economic scenario generator.

The procedure essentially consists of computing $Q(\cdot)$ and $U(\cdot)$ in each financial scenario by solving a system of (stochastic) differential equations. The involved part is to evaluate the differentials. The procedure looks as follows. For each financial scenario $k=1,\ldots,N$:
\begin{itemize}
\item Initialize with $Q^k(0) = 0$ and $U^k(0)=u_0$. 
\item Apply a numerical algorithm to solve the coupled (stochastic) differential equation systems for $Q^k(\cdot)$ and $U^k(\cdot)$ from~\eqref{eq:Qtilde} and~\eqref{eq:dU}, respectively.
\begin{itemize}
\item Evaluating the differentials at time $t$ involves the mappings $(\widetilde{\delta}_0,\widetilde{\delta}_2,\eta)$ from \eqref{eq:eta} and~\eqref{eq:deltaform}. By inspection of the differentials and these mappings, we see that we require the shape of the insurance business
\begin{align*}
\mathcal{I}^k(t)=\left(U^k(t),\bar{A}^{g,k}(t,\mathrm{d}s),\bar{V}^{\star,k}(t)\right)\!,
\end{align*}
the expected bonus cash flow $a^{b,k}(0,t) = Q^k(t)a^\dagger(0,t)$, cf.~\eqref{eq:BonusCF_stateIndep}, as well as the input. Computation of $\bar{A}^{g,k}(t,\mathrm{d}s)$ and $\bar{V}^{\star,k}(t)$ is achieved via Corollary~\ref{col:state_indep_projmarket} and Corollary~\ref{col:state_indep_projtechnical}.
\end{itemize}  
\item We emphasize that as part of evaluating the differentials we computed the expected bonus cash flow $a^{b,k}(0,\cdot)$.
\end{itemize}
The procedure completes by computing the market value of bonus payments $V^b(0)$ via
\begin{align*}
 V^b(0) &\approx
\frac{1}{N}\sum_{k=1}^N \int_0^n e^{-\int_0^{t} r^k(v) \md v} a^{b,k}(0,t)\md t
\end{align*} 
using an algorithm for numerical integration.

Note that in comparison with the procedure of Subsection~\ref{subsec:proj_implement}, the expected unit bonus cash flows $a_j^\dagger(t,\cdot)$, $j\in\J$, have only to be precomputed for $j=z_0$ and $t=0$. This leads to a speedup. Additionally, the procedure itself does not depend on the market basis for the specific insured (except potentially through the mappings $\widetilde{\delta}_0$, $\widetilde{\delta}_2$, and $\eta$). These are the primary practical advantages that are gained by strengthening the requirements on the dividend strategy to~\eqref{eq:deltaform}--\eqref{eq:deltaform2}.

\section{Numerical example} \label{sec:num_ex}

In this section, we illustrate the methods presented in the previous sections via a numerical example intended to show how our methods and results can be applied in practice. The predetermined payments, technical basis, and market basis are based on the numerical example in~\cite{buchardt2015}; our extension consists of the inclusion of financial risk and bonus payments. The numerical example aims at illustrating similarities and differences between the state-dependent scenario-based projection model and a state-independent approximation in the spirit of~\cite[][Section 4]{JensenSchomacker2015}, cf.\ Example~\ref{ex:second_order_int}.

\subsection{Setup} \label{subsec:num_ex_setup}

The state of the insured is modeled in an eight-state disability model with policyholder behavior as depicted in Figure~ \ref{fig:extmm_numerical}. We consider a male who is $40$ years old today. His retirement age is taken to be $65$, and his predetermined payments consist of:
\begin{itemize}
\item A disability annuity of rate $100000$ per year while disabled, but only until retirement, i.e.\ age $65$.
\item A life annuity of rate $100000$ per year while alive and non-lapsed (corresponding to states $0$, $1$, $4$, and $5$), but only from retirement, i.e.\ age $65$.
\item Premium payments of rate $46409.96$ per year while active, but only until retirement, i.e.\ age $65$.
\end{itemize}
The maximal contract time is set equal to $70$, i.e.\ $n=70$, corresponding to a maximal age of $40+70=110$ years.
\begin{figure}[htb]
\centering
\scalebox{0.9}{%
\begin{tikzpicture}[node distance=2em and 0em]
\node[punkt]                                       (1)    {disabled};
\node[anchor=north east, at=(1.north east)]               {$1$};
\node[punkt, left = 40mm of 1]                     (0)    {active};
\node[anchor=north east, at=(0.north east)]               {$0$};
\node[draw = none, fill = none, left = 20 mm of 1] (test) {};
\node[punkt, below = 20mm of test]                 (2)    {dead};
\node[anchor=north east, at=(2.north east)]               {$2$};
\node[punkt, left = 20mm of 0]                     (3)    {surrender};
\node[anchor=north east, at=(3.north east)]               {$3$};
\node[punkt, below = 40mm of 0]                    (4)    {active \\ free policy};
\node[anchor=north east, at=(4.north east)]               {$4$};
\node[punkt, right = 40mm of 4]                    (5)    {disabled \\ free policy};
\node[anchor=north east, at=(5.north east)]               {$5$};
\node[draw = none, fill = none, left = 20 mm of 5] (df)   {};
\node[punkt, below = 20mm of df]                   (6)    {dead \\ free policy};
\node[anchor=north east, at=(6.north east)]               {$6$};
\node[punkt, left = 20mm of 4]                     (7)    {surrender \\ free policy}; 
\node[anchor=north east, at=(7.north east)]               {$7$};
    
\path
(0) edge [pil, bend left = 15] node [above]  {$\mu_{01}$} (1)
(1) edge [pil]          node [below right] {$\mu_{12}$} (2)
(0) edge [pil]          node [below left] {$\mu_{02}$} (2)
(1) edge [pil, bend left = 15] node [below]  {$\mu_{10}$} (0)
(0) edge [pil]                  node [below left]   {$\mu_{04}$}    (4)
(0) edge [pil]                  node [above]        {$\mu_{03}$}    (3)
(4) edge [pil]                  node [above]        {$\mu_{47}$}    (7)
(4) edge [pil, bend left = 15] node [above]  {$\mu_{45}$} (5)
(5) edge [pil]          node [below right] {$\mu_{56}$} (6)
(4) edge [pil]          node [below left] {$\mu_{46}$} (6)
(5) edge [pil, bend left = 15] node [below]  {$\mu_{54}$} (4)
;
\end{tikzpicture}
}
\caption{Disability model with policyholder behavior. The state-space is decomposed according to $\J = \Jp\cup\Jf$ with $\Jp = \tuborg{0,\ldots,3}$ and $\Jf = \tuborg{4,\ldots,7}$.}
\label{fig:extmm_numerical}
\end{figure}
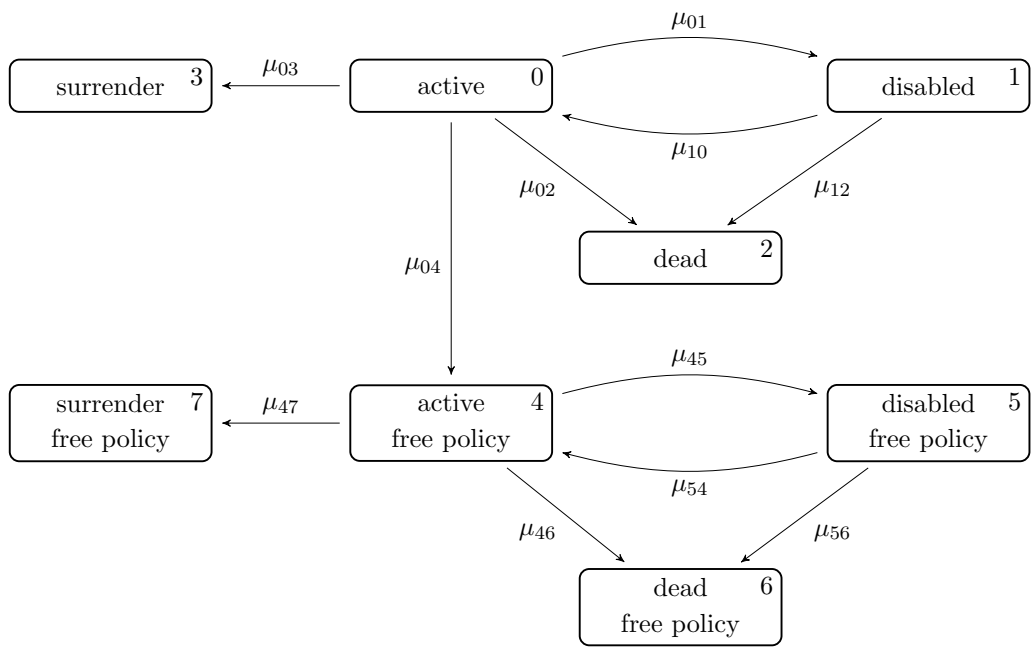
We note that the predetermined payments are actuarially fair in the sense that the equivalence principle is satisfied on the technical basis. The technical basis takes the following form:
\begin{align*}
r^\star(s) &= 0.01, \\
\mu^\star_{01}(s) &= \left(0.0004+10^{4.54+0.06(s+40)-10} \right)\!\mathds{1}_{(s\leq 25)}, \\
\mu_{10}^\star(s) &= \left(2.0058 e^{-0.117(s+40)}\right)\!\mathds{1}_{(s\leq 25)}, \\
\mu_{02}^\star(s) &= 0.0005+10^{5.88+0.038(s+40)-10}, \\
\mu_{12}^\star(s) &=  \mu_{02}^\star(s)\!\left(1+\mathds{1}_{(s\leq 25)}\right)\!. 
\end{align*}
The technical basis and the aforementioned predetermined payments further determine the surrender payments and the free policy factor, which we do not explicitly state, cf.\ Subsection~\ref{subsec:unit_val}.

The market basis takes the following form:
\begin{align*}
\mu_{02}(\cdot)&: \text{The 2012 edition of the Danish FSA's longevity benchmark}, \\ 
\mu_{01}(s) &= 10^{5.662015+0.033462(s+40)-10}\mathds{1}_{(s\leq 25)}\!, \\
\mu_{10}(s) &= 4.0116 e^{-0.117(s+40)}\mathds{1}_{(s\leq 25)}, \\
\mu_{12}(s) &= \left(0.010339 + 10^{5.070927+0.05049(s+40)-10}\right)\!\mathds{1}_{(s\leq 25)} + \mu_{02}(s)\mathds{1}_{(s>25)}, \\
\mu_{03}(s) &= \left(0.06 - 0.002s \right)\!\mathds{1}_{(s\leq 25)}, \\
\mu_{04}(s) &= 0.05\mathds{1}_{(s\leq 25)},
\end{align*} 
with $\mu_{jk} = \mu_{(j+4)(k+4)}$ for $j,k\in \J^f$, $j\neq k$. 
 
We now deviate from \cite{buchardt2015} by introducing a bond market, so that the risky asset with price process $S_1$ corresponds to a zero-coupon bond with expiry $n=70$. The short rate follows a Vasicek model, so that the dynamics are given by
\begin{align*}
\md r(t) &= \left(\beta-\alpha\cdot  r(t)\right)\!\!\md t +\sigma\!\md W(t), \quad r(0) = 0.01,  \\
\md S_1(t) &= r(t)S_1(t)\!\md t - \sigma \psi(t,n)S_1(t)\!\md W(t), \quad S_1(n) = 1, 
\end{align*}
where $\psi(t,n) = \left(1-e^{-\alpha (n-t)}\right)\!/\alpha$. The parameters $\beta$, $\alpha$, and $\sigma$ may be found in Table~\ref{table:par}. They are taken from the numerical example in~\cite{skatogomk} and yield a mean reversion to about $0.043$.

Regarding bonus payments, we consider the case where all dividends are used only to buy additional life annuity benefits, so that in particular the rate of the disability annuity is kept fixed throughout the entire contract period. Consequently, the unit bonus payments $B^\dagger$ are determined according to
\begin{align*}
b_0^\dagger(t) = b_1^\dagger(t) = 100000\cdot \mathds{1}_{(t\geq 25)}. 
\end{align*}  
 Regarding the controls $(\delta,\eta)$, we consider the dividend strategy $\delta$ introduced in Example \ref{ex:rente} with the second order interest rate \eqref{eq:r_delta_ex}, and investment strategy $\eta$ given by:
\begin{align*}
\eta(t) &= \frac{\int_t^T \psi(t,s)e^{-\int_t^s f(t,v)\md v}\bar{A}^g(t,\!\md s) }{\psi(t,T)S_1(t)}.
\end{align*}
The investment strategy is chosen such that it hedges the interest rate risk of the guaranteed cash flows $\bar{A}^g$ on a portfolio level.

The values of the parameters for the short rate model and the dividend strategy are shown in Table~\ref{table:par}. 
\begin{table}[h!]
\centering
\begin{tabular}{c|cccc}
Parameter & $\beta$ & $\alpha$ & $\sigma$ & $\kappa$ \\ \hline Value  
 & 0.007006001 
 & 0.162953 
 & 0.015384
 & 0.2
\end{tabular}
\caption{Parameters for the short rate model and dividend strategy. The parameters of the former are taken from the numerical example in~\cite{skatogomk}, which provides a mean reversion of the short rate to $\beta/\alpha \approx 0.043$.}
\label{table:par}
\end{table}

\subsection{Results and discussion} \label{subsec:num_ex_results_discussion}

The inputs as described in Subsection~\ref{subsec:proj_implement} and Subsection~\ref{subsec:implement_indep} are computed in classic fashion using standard numerical methods. Next, we carry out the state-dependent numerical procedure outlined in Subsection~\ref{subsec:proj_implement} as well as the state-independent procedure presented in Subsection~\ref{subsec:implement_indep} to determine the time zero market value of bonus payments $V^b(0)$; the computations are based $N=10000$ financial scenarios and Euler-Maruyama discretizations with step length $0.01$ years. In the latter procedure, we use the dividend strategy presented in Example~\ref{ex:second_order_int} with $r^\delta$ on the same form as above. The results are presented in Table \ref{tab:res} along with the market value of predetermined payments $V^\circ(0)$. 
\begin{table}[h!]
\centering
\begin{tabular}{cccc}
   $V^\circ(0)$  & \multicolumn{3}{c}{Time zero market value of bonus payments $V^b(0)$} \\ 
   & State-dependent & State-independent & Relative difference \\ \hline
  -72582  & 72661 & 72663 & -0.00201\%
\end{tabular}
\caption{Time zero market values from both the state-dependent and state-independent implementation. The relative difference lies within the margin of numerical error.}
\label{tab:res}
\end{table}
We see that the two implementations produce identical results, in the sense that the difference is within the margin of numerical error, for this product design and set of parameter values.

To show what is going on behind the scenes, we investigate in a bit more detail the inner workings of the state-independent and state-dependent numerical procedures. To this end, we fix the financial scenario presented in Figure~\ref{fig:scenario}.
\begin{figure}[htb]
\centering
\includegraphics[width=0.9\textwidth]{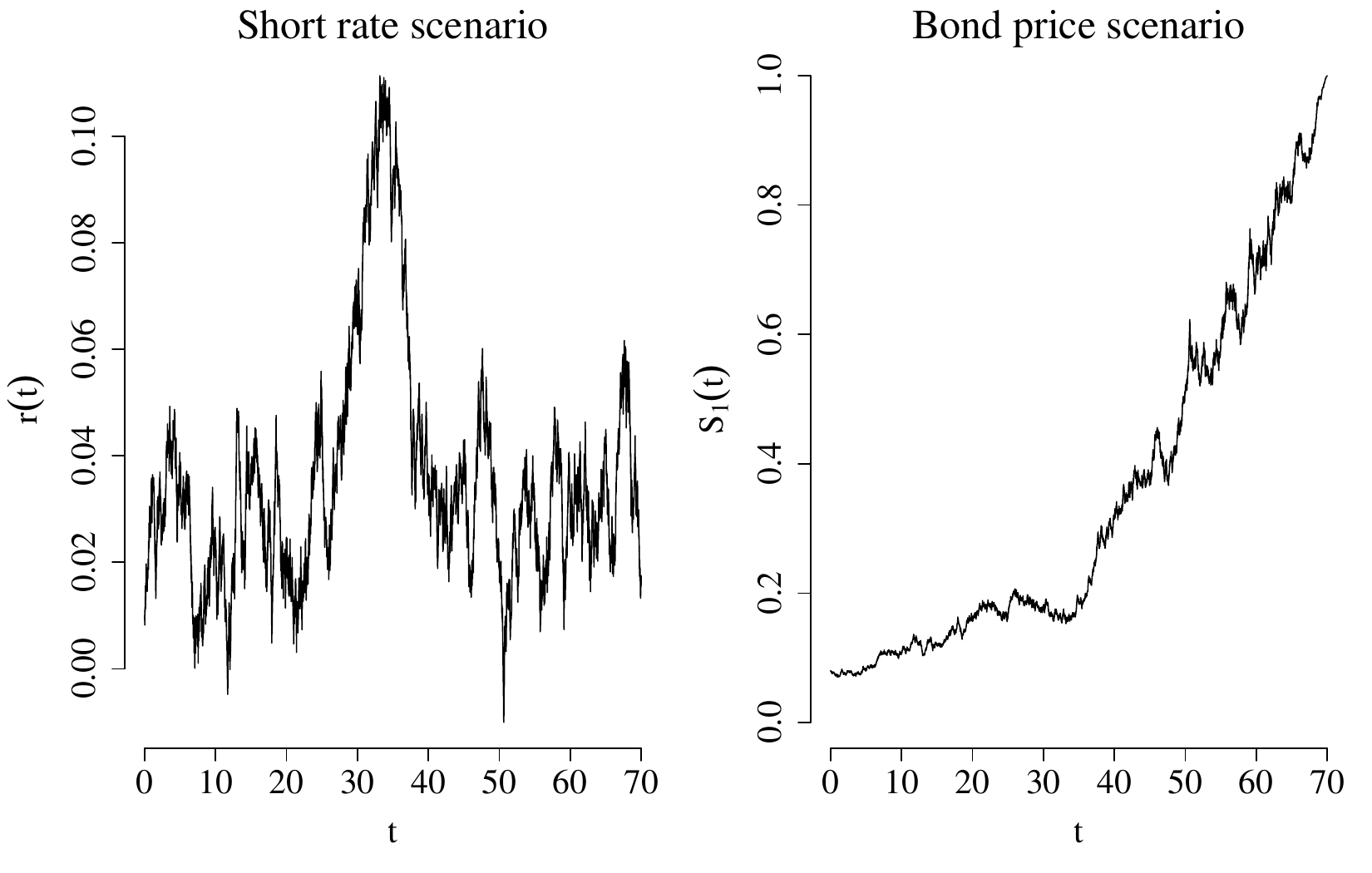}
\caption{Short rate and bond price corresponding to a single financial scenario.}
\label{fig:scenario}
\end{figure}

In Figure~\ref{fig:pQ_comparison}, we compare $t\mapsto p^Q_{z_0j}(0,t)$, $j\in \J$, of the state-dependent implementation to $t\mapsto Q(t)p_{z_0j}(0,t)$ of the state-independent implementation. With Proposition~\ref{prop:FDB} in mind, we also compare differences in bonus cash flows across states. To be precise, Figure~\ref{fig:state_wise_bonus_cf} contains a comparison of state-wise bonus cash flows, i.e.\ here we compare
\begin{align*}
t\mapsto p^Q_{z_0j}(0,t)\Big(b_j^\dagger(t) + \sum_{k\in \J \atop k\neq j} b_{jk}^\dagger(t)\mu_{jk}(t)\Big)
\end{align*}
of the state-dependent implementation to
\begin{align*}
t\mapsto Q(t)p_{z_0j}(0,t)\Big(b_j^\dagger(t) + \sum_{k\in \J \atop k\neq j} b_{jk}^\dagger(t)\mu_{jk}(t)\Big)
\end{align*}
of the state-independent implementation.
\begin{figure}[htb]
\centering
\includegraphics[width=0.9\textwidth]{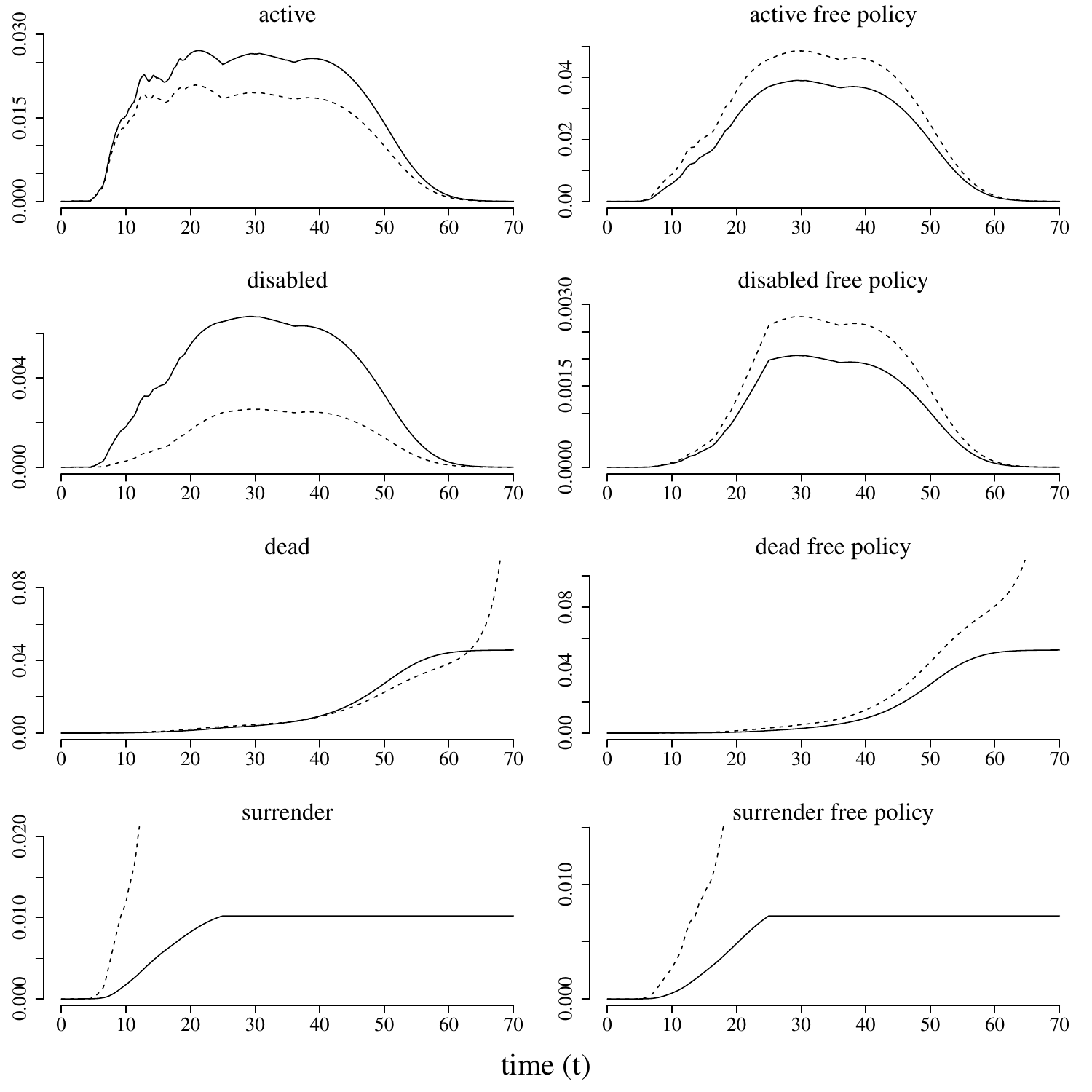}
\caption{$Q$-modified transition probabilities for a single financial scenario in the state-dependent implementation (solid line) and in the state-independent implementation (dashed line).}
\label{fig:pQ_comparison}
\end{figure}
\begin{figure}[htb]
\centering
\includegraphics[width=0.9\textwidth]{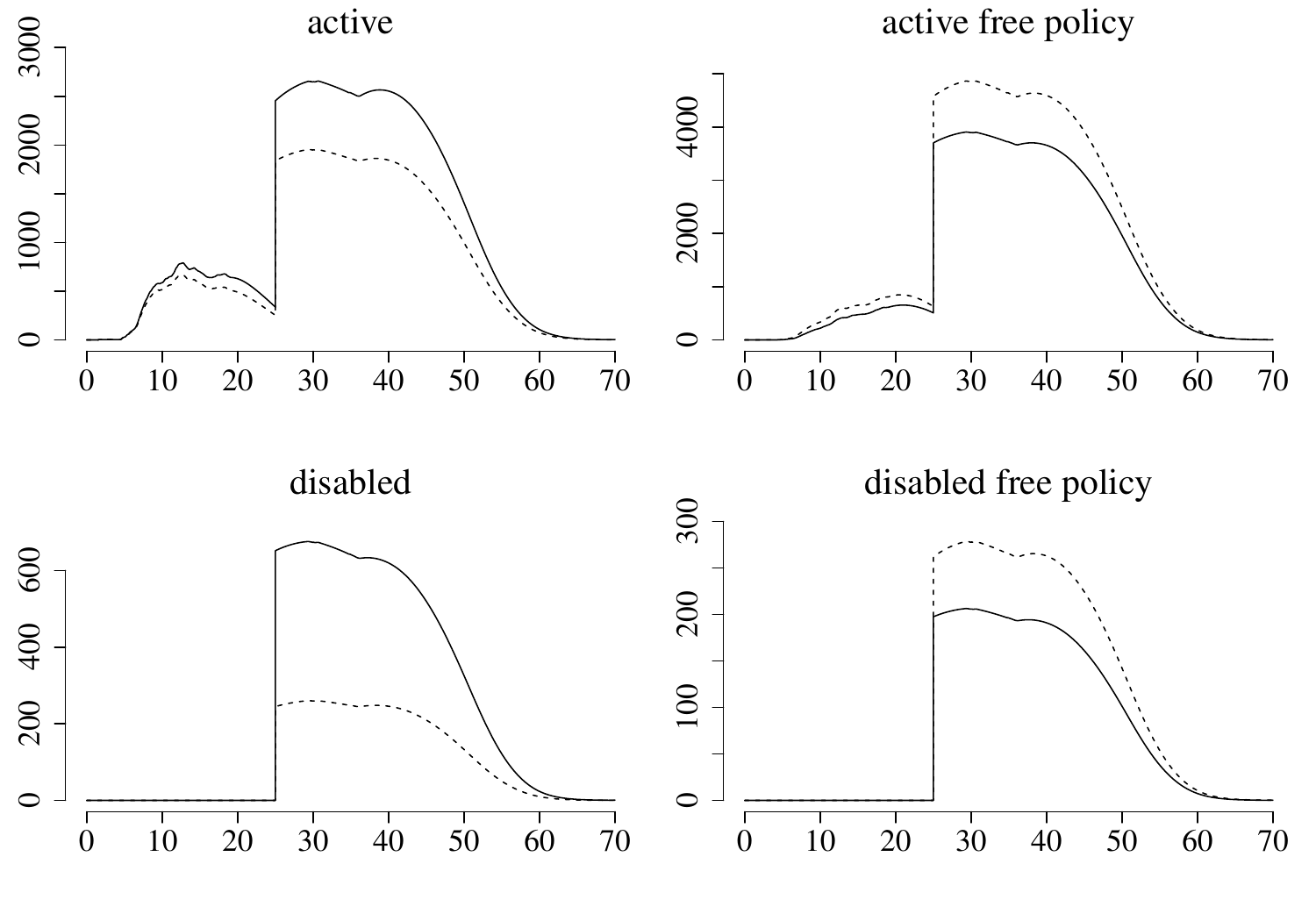}
\caption{State-wise bonus cash flows for a single financial scenario in the state-dependent implementation (solid line) and in the state-independent implementation (dashed line).}
\label{fig:state_wise_bonus_cf}
\end{figure}
In general, the two numerical procedures lead to fundamentally different intermediate quantities. The state-independent implementation both overestimates and underestimates -- depending on the state -- intermediate state-specific quantities compared to the state-dependent implementation. But intriguingly the aggregation over states cancels these differences, resulting in the same value for $V^b(0)$.

To offer a different point of view based on Remark~\ref{rmk:approx_expression} and Example~\ref{ex:second_order_int}, we take a closer look at the value of the dividend payments between implementations. To elaborate, we compare
\begin{align*}
t\mapsto \E\!\left[\left. \big(r^\delta(t)-r^\star(t)\big) V^\star(t)\frac{V^\dagger_{Z(t)}(t)}{V_{Z(t)}^{\star,\dagger}(t)}\, \right|  \mathcal{F}^S(t)\right]
\end{align*}
of the state-dependent implementation to
\begin{align*}
t\mapsto \E\!\left[\left. \big(r^\delta(t)-r^\star(t)\big) \bar{V}^{\star}(t) \frac{V_{Z(t)}^\dagger(t)}{\bar{V}^{\star,\dagger}(t)} \, \right| \mathcal{F}^S(t)\right]
\end{align*}
of the state-independent implementation. The absolute relative differences are less than $1 \%$, which indicates that it is the aggregation over states and not times that cancels the intermediate differences and gives rise to essentially identical values for $V^b(0)$.

\section{Final remarks}\label{sec:outlook}

In this section, we compare our methodology and results with recent advances in the literature and discuss possible extension in demand by practitioners. Subsection~\ref{subsec:compare} contains comparisons with~\cite{Lollike2020,Falden2020,JensenSchomacker2015}, while the inclusion of both duration effects (so-called semi-Markovianity) and the bonus scheme \textit{consolidation} is the focal point of Subsection~\ref{subsec:extend}.

\subsection{Comparison with recent advances in the literature}\label{subsec:compare}

In~\cite{Lollike2020} and the follow-up paper~\cite{Falden2020}, where the methods and results of the former are generalized to allow for surrender and free policy conversion, primary attention is given to the derivation of differential equations for quantities such as
\begin{align*}
\E\!\left[\left. \mathds{1}_{(Z(t) = j)} V^\star(t)\, \right| \mathcal{F}^S(t)\right]\!.
\end{align*}
Since $V^\star = V^{\star,\circ} + Q\cdot V^{\star,\dagger}$, we find that $t \mapsto \mathds{1}_{(Z(t) = j)} V^\star(t)$ is an affine function of $t \mapsto \mathds{1}_{(Z(t) = j)} Q(t)$. Thus disregarding free policy conversion, we see a direct link between the differential equations derived in~\cite{Lollike2020,Falden2020} and those of Theorem~\ref{thm:diffQ}. For these results suitable affinity of the dividend strategy is a key assumption.

The inclusion of the policyholder option of free policy conversion adds an additional layer of complexity. We assumed the unit bonus payment stream $B^\dagger$ to be unaffected by the free policy option, which leads to the total payment stream given by~\eqref{eq:B}. No such assumption is made in~\cite{Falden2020}, which leads to more involved payment streams, although by setting $B^{\dagger} = B^{\circ,+}$, our payment stream equals that of~\cite[][Subsection 5.2, cf.~(11)--(12)]{Falden2020}.

We consider some key concepts and provide practical insights that are not within the scope of~\cite{Lollike2020,Falden2020}. We explicitly include financial risk, which serves as a good starting point for the extension to doubly stochastic models with dependence between the financial market and the stochastic transition rates. Moreover, we identify and discuss the theoretical and practical challenges arising from the fact that the dividend strategy depends on the shape of the insurance business. Furthermore, we provide ready-to-implement numerical schemes for the computation of the market value of bonus payments. Finally, we discuss potential simplifications arising when the number of additional benefits is (approximated to be) $\mathcal{F}^S$-adapted -- the state-independent case, which might be of particular interest to practitioners.

As discussed previously, the projection model described in~\cite[][Section 4]{JensenSchomacker2015} appears to be conceptually very close to exactly our state-independent model, cf.\ Subsection~\ref{subsec:class_div} and Example~\ref{ex:second_order_int}. Consequently, we believe that our presentation among other things serves to formalize and generalize the pragmatic approach found in~\cite{JensenSchomacker2015} and, correspondingly, aims at bridging the gap between the methods and results found in~\cite{Lollike2020,Falden2020} and~\cite{JensenSchomacker2015}.

\subsection{State-independent approximations}\label{subsec:SI_approx}

The dividend strategies considered in the numerical example of Section~\ref{sec:num_ex} are to some extent absorbing in nature, meaning that they over time broadly speaking distribute the available assets to the insured. This is consistent with the fact that both the state-dependent and state-independent implementation produce a time zero market value of bonus payments identical to the available assets.

In addition to the numerical results presented in Section~\ref{sec:num_ex}, we have examined other product designs and sets of parameter values. Without reporting every detail here and now, this did not immediately produce large differences in the time zero market value of bonus payments between implementations. To uncover why this is the case, further theoretical and numerical studies are required.

\subsection{Extensions}\label{subsec:extend}

In both theory and practice, the generalization to so-called semi-Markovian models introducing duration dependence in the transition rates and payments is popular and impactful, cf.~\cite{hoem72,helwich,christiansen2012,BuchardtMollerSchmidt}. We believe that the methods we use here can easily be adapted to semi-Markovian models.

The increase in numerical speed from the general case to the state-independent case is increasing in the complexity of the intertemporal dependence structure, which can be seen as follows. Referring to Subsection~\ref{subsec:proj_implement} and Subsection~\ref{subsec:implement_indep}, the general projection model requires as input the expected unit bonus cash flows evaluated on a two-dimensional time grid, while evaluation on a one-dimensional time grid suffices for the state-independent model. When including duration effects, the complexity increases, which ought to entail a four-dimensional time/duration grid for the expected unit bonus cash flows in general projections and a two-dimensional time/duration grid in state-independent projections. The gain in numerical speed by assuming the state-independent special case is thus far greater in the semi-Markovian model compared to the Markovian model.

In Denmark, the bonus scheme known simply as \textit{consolidation} (in Danish: \textit{styrkelse}) sees widespread use in practice, cf.~\cite[][Subsection 4.1]{JensenSchomacker2015}. Consolidation involves two technical bases: a low (more prudent) basis and a high (less prudent) basis. At the onset of the contract, the predetermined payments, i.e.\ the payments guaranteed at time zero, satisfy an equivalence principle for which some payments are valuated on the high technical basis and the remaining payments are valuated on the low technical basis. Dividends are then used to shift these payments from the high to the low basis while upholding the relevant equivalence principle. Typically consolidation is combined with the bonus scheme additional benefits in the following manner. When all predetermined payments have been shifted to the low technical basis, future dividends are used to buy additional benefits. This ruins a key affinity assumption, which increases the complexity significantly. In particular, an extension of Theorem~\ref{thm:diffQ} appears to require more sophisticated methods. In the state-independent case, the assumption of affinity is not required, cf.\ Remark~\ref{rmk:non_affine_stateindep}. Consequently, we believe that it is straightforward to extend the state-independent projection model to include consolidation in combination with additional benefits.
 
\section*{Acknowledgments and declarations of interest}

Our work was at first motivated by unpublished notes by Thomas Møller. We thank Thomas Møller for fruitful discussions and for sharing his insights with us. We are also grateful to Ann-Sophie Buchardt for her assistance with the delightful figures.

Christian Furrer’s research has partly been funded by the Innovation Fund Denmark (IFD) under File No.\ 7038-00007. The authors declare no competing interests.



\vspace{4mm}
{\large
\textsc{Jamaal Ahmad} \\[2mm]
\textit{Department of Mathematical Sciences, University of Copenhagen \\[2mm]
Universitetsparken 5, DK-2100 Copenhagen \O, Denmark, \\[0mm]
E/04, Building: 04.3.20 \\[2mm]
E-Mail: \href{mailto:jamaal@math.ku.dk}{jamaal@math.ku.dk}}
}

\end{document}